\documentclass[journal]{IEEEtran}
\usepackage{amsmath}
\usepackage{amssymb}
\usepackage{amsfonts}
\usepackage{graphicx}
\usepackage{epsfig}
\usepackage{subfigure}
\usepackage{psfrag}
\usepackage{color}
\usepackage{cite}

\begin{document}

\title{Wireless Information and Power Transfer: Architecture
Design and Rate-Energy Tradeoff}

\author{Xun Zhou, Rui Zhang, and Chin Keong Ho
\thanks{This paper has been presented in part at IEEE Global Communications
Conference (Globecom), December 3-7, 2012, California, USA.}
\thanks{X. Zhou is with the Department of
Electrical and Computer Engineering, National University of Singapore
(e-mail: xunzhou@nus.edu.sg).}
\thanks{R. Zhang is with the Department
of Electrical and Computer Engineering,
National University of Singapore (e-mail: elezhang@nus.edu.sg).
He is also with the Institute for Infocomm Research, A*STAR, Singapore.}
\thanks{C. K. Ho is with the Institute for Infocomm Research, A*STAR,
Singapore (e-mail: hock@i2r.a-star.edu.sg).}}

\maketitle

\begin{abstract}
Simultaneous information and power transfer over the wireless channels potentially
offers great convenience to mobile users. Yet practical receiver designs impose
technical constraints on its hardware realization, as practical circuits for
harvesting energy from radio signals are not yet able to decode the carried
information directly. To make theoretical progress, we propose a general receiver
operation, namely, \emph{dynamic power splitting} (DPS), which splits the received
signal with adjustable power ratio for energy harvesting and information decoding,
separately. Three special cases of DPS, namely, \emph{time switching} (TS), \emph{static power splitting}
(SPS) and \emph{on-off power splitting} (OPS) are investigated. The TS and SPS schemes can be
treated as special cases of OPS. Moreover, we propose two types of practical receiver
architectures, namely, \emph{separated} versus \emph{integrated} information and energy receivers.
The integrated receiver integrates the front-end components of the separated receiver, thus achieving
a smaller form factor. The rate-energy tradeoff for the two architectures are characterized by a so-called
\emph{rate-energy} (R-E) region. The optimal transmission strategy is derived to achieve different
rate-energy tradeoffs. With receiver circuit power consumption taken into account, it is shown that
the OPS scheme is optimal for both receivers. For the ideal case when the receiver circuit does not
consume power, the SPS scheme is optimal for both receivers. In addition, we study the performance
for the two types of receivers under a realistic system setup that employs practical modulation.
Our results provide useful insights to the optimal practical receiver design for simultaneous wireless information
and power transfer (SWIPT).
\end{abstract}


\IEEEpeerreviewmaketitle

\setlength{\baselineskip}{1.0\baselineskip}

\newtheorem{definition}{\underline{Definition}}[section]
\newtheorem{fact}{Fact}
\newtheorem{assumption}{Assumption}
\newtheorem{theorem}{\underline{Theorem}}[section]
\newtheorem{lemma}{\underline{Lemma}}[section]
\newtheorem{corollary}{\underline{Corollary}}[section]
\newtheorem{proposition}{\underline{Proposition}}[section]
\newtheorem{example}{\underline{Example}}[section]
\newtheorem{remark}{\underline{Remark}}[section]
\newtheorem{algorithm}{\underline{Algorithm}}[section]

\newcommand{\mv}[1]{\mbox{\boldmath{$ #1 $}}}

\newcommand{\W}{2\pi ft}                
\newcommand{\Xt}{x(t)}                  
\newcommand{\Ybt}{y_{\rm b}(t)}         
\newcommand{\Yo}{\hat y[k]}             
\newcommand{\Yonew}{\hat y[k]}          

\newcommand{\Nax}{\tilde n_{\rm A}(t)}                
\newcommand{\Na}{n_{\rm A}(t)}                        
\newcommand{\Ncov}{n_{\rm cov}(t)}                    

\newcommand{\Va}{\sigma^2_{\rm A}}
\newcommand{\Vcov}{\sigma^2_{\rm cov}}
\newcommand{\Vrec}{\sigma^2_{\rm rec}}                  
\newcommand{\Srec}{\sigma_{\rm rec}}                    

\newcommand{\Rp}{\sqrt{hP}}         
\newcommand{\Pcs}{P_{\rm S}}        
\newcommand{\Pci}{P_{\rm I}}        
\newcommand{\Rs}{R_1^\ast}          
\newcommand{\Ri}{R_2^\ast}          

\section{Introduction}

Harvesting energy from the environment is a promising approach to prolong
the lifetime of energy constrained wireless networks.
Among other renewable energy sources such as solar and wind, background
radio-frequency (RF) signals radiated by ambient transmitters can
be a viable new source for wireless power transfer (WPT).
On the other hand, RF signals have been
widely used as a vehicle for wireless information transmission (WIT).
Simultaneous wireless information and power transfer (SWIPT) becomes appealing
since it realizes both useful utilizations of RF signals at the same time,
and thus potentially offers great convenience to mobile users.

Simultaneous information and power transfer over the wireless channels has been studied in
\cite{Varshney,Sahai,Liu,Zhang,Xiang,Chalise,Fouladgar,Huang,Lee}.
Varshney first proposed the idea of
transmitting information and energy simultaneously in \cite{Varshney}. A capacity-energy
function was proposed to characterize the fundamental performance tradeoff for
simultaneous information and power transfer.
In \cite{Sahai}, Grover and Sahai extended the work in \cite{Varshney} to frequency-selective channels
with additive white Gaussian noise (AWGN). It was shown in \cite{Sahai} that a non-trivial tradeoff
exists for information transfer versus energy transfer via power allocation.
Wireless information and power transfer subject to co-channel interference was studied in \cite{Liu},
in which optimal designs to achieve different outage-energy tradeoffs as well as rate-energy tradeoffs are derived.
Different from the traditional view of taking interference as an undesired factor that jeopardizes the
wireless channel capacity, in \cite{Liu} interference was utilized as a source for energy harvesting.
Unlike \cite{Varshney,Sahai,Liu}, which considered point-to-point single-antenna
transmission, \cite{Zhang,Xiang,Chalise} considered multiple-input multiple-output
(MIMO) systems for SWIPT.
In particular, \cite{Zhang} studied the performance limits of a three-node MIMO broadcasting system,
where one receiver harvests energy and another receiver decodes information
from the signals sent by a common transmitter.
\cite{Xiang} extended the work in \cite{Zhang}
by considering imperfect channel state information (CSI) at the transmitter.
MIMO relay systems involving an energy harvesting receiver were studied in \cite{Chalise},
in which the joint optimal source and relay precoders were designed to achieve different
tradeoffs between the energy transfer and the information rates.
SWIPT for multi-user systems was studied in \cite{Fouladgar}.
It was shown in \cite{Fouladgar} that for multiple access channels with a received energy constraint,
time-sharing is necessary to achieve the maximum sum-rate when the received energy constraint
is sufficiently large; while for the multi-hop channel with a harvesting relay, the transmission
strategy depends on the quality of the second link.
Networks that involve wireless power transfer were studied in \cite{Huang,Lee}.
In \cite{Huang}, the authors studied a hybrid network which overlays an uplink cellular network
with randomly deployed power beacons that charge mobiles wirelessly.
Under an outage constraint on the data links, the tradeoffs between the network parameters
were derived. In \cite{Lee}, the authors investigated a cognitive radio network powered by opportunistic
wireless energy harvesting, where mobiles from the secondary network either harvest
energy from nearby transmitters in a primary network, or transmit information if the
primary transmitters are far away. Under an outage constraint for coexisting networks,
the throughput of the secondary network was maximized.

Despite the recent interest in SWIPT, there remains
two key challenges for practical implementations.
First, it is assumed in \cite{Varshney,Sahai} that the receiver is able to observe
and extract power simultaneously from the same received signal. However, this assumption may not
hold in practice, as practical circuits for harvesting energy from radio signals are
not yet able to decode the carried information directly. Due to this potential limitation,
the results in \cite{Varshney,Sahai} actually provided only optimistic performance bounds.
To coordinate WIT and WPT at the receiver side, two practical
schemes, namely, time switching (TS) and static power splitting (SPS), were proposed in \cite{Zhang}.
Second, the conventional information receiver architecture designed for WIT may not be optimal for
SWIPT, due to the fact that WIT and WPT operate with very
different power sensitivity at the receiver (e.g., -10dBm for energy receivers versus
-60dBm for information receivers). Thus, for a system that involves both WIT and WPT, the receiver
architecture should be optimized for WPT. In addition, circuit power consumed by information
decoding becomes a significant design issue for simultaneous information and power transfer,
since the circuit power reduces the net harvested energy that can be stored in the battery for future use. In particular, the active mixers used in conventional information receiver
for RF to baseband conversion are substantially power-consuming. It thus motivates us to propose new receiver
architectures which consume less power by avoiding the use of active devices.

In this paper, we study practical receiver designs for a point-to-point wireless
link with simultaneous information and power transfer (see Fig. \ref{fig:system model}).
We generalize the TS and SPS schemes proposed in \cite{Zhang} to a general receiver operation
scheme, namely, {\it dynamic power splitting} (DPS), by which the signal is dynamically split into
two streams with arbitrary power ratio over time \cite{Powerdivider,Agilent}.
Besides TS and SPS, another special case of the DPS scheme, namely, {\it on-off power splitting} (OPS)
is also investigated.
The TS and SPS schemes can be treated as special cases of the OPS scheme.
Based on DPS, we first consider the {\it separated} receiver architecture (see Fig. \ref{fig:SepRx})
based on a conventional information-decoding architecture. In this architecture,
the received signal by the antenna is split
into two signal streams in the RF band, which are then separately fed to the conventional
energy receiver and information receiver for harvesting energy and decoding information, respectively.
Note that the receivers in \cite{Zhang} implicitly employ the separated receiver architecture.
Instead, we propose an {\it integrated} receiver architecture (see Fig. \ref{fig:IntRx}), in which we integrate
the information decoding and the energy harvesting circuits. In this architecture,
the active RF band to baseband conversion in conventional information decoding is replaced
by a passive rectifier operation, which is conventionally used only for energy harvesting.
By providing a dual use of the rectifier, the energy cost for information decoding is reduced significantly.

The rate-energy performances for the two proposed receivers are further characterized by
a so-called rate-energy (R-E) region. With receiver circuit power consumption taken into account,
it is shown that the OPS scheme is optimal for both receivers. For the ideal case when the
consumed power at the receiver is negligible, the SPS scheme is optimal for both receivers.
Finally, the performance for the two receivers are compared under a realistic system setup that employs practical
modulation. The results show that for a system with zero-net-energy consumption,
the integrated receiver achieves more rate than separated receiver at sufficiently small transmission distance.

The rest of this paper is organized as follows: Section \ref{sec:system model} presents the system model. Section
\ref{sec:architecture} proposes the two receiver architectures. Section \ref{sec:separated receivers} and
Section \ref{sec:integrated receivers} study the rate-energy performance for the separated and integrated receivers,
respectively. Section \ref{sec:circuit power} extends the results in Sections \ref{sec:separated receivers}
and \ref{sec:integrated receivers} to the case with receiver circuit power taken into consideration.
Section \ref{sec:practical modulation} studies the performance for the two types of receivers under
a realistic system setup. Finally, Section \ref{sec:conclusion} concludes the paper.

\section{System Model}\label{sec:system model}

\subsection{Channel Model}

\begin{figure}
\begin{center}
\scalebox{0.48}{\includegraphics{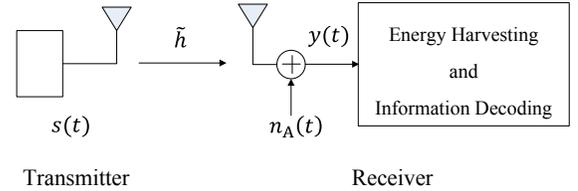}}
\end{center}
\caption{System model.}
\label{fig:system model}
\end{figure}

As shown in Fig. \ref{fig:system model}, this paper studies a
point-to-point wireless link with simultaneous information and power
transfer. Both the transmitter and receiver are
equipped with one antenna. At the transmitter side, the
complex baseband signal is expressed as $x(t)=A(t)e^{j\phi(t)}$,
where $A(t)$ and $\phi(t)$ denote the amplitude and the phase of
$x(t)$, respectively. It is assumed that $x(t)$ is a narrow-band
signal with bandwidth of $B$ Hz, and $\mathbb{E}[|x(t)|^2]=1$, where
$\mathbb{E}[\cdot]$ and $|\cdot|$ denote the statistical expectation
and the absolute value, respectively. The transmitted RF band signal
is then given by $s(t)=\sqrt{2P} A(t)\cos\left( \W+\phi(t)
\right)=\sqrt{2P}\Re\{ x(t) e^{j\W} \}$, where $P$ is the average
transmit power, i.e., $\mathbb{E}[s^2(t)]=P$, $f$ is the carrier
frequency, and $\Re\{\cdot\}$ denotes the real part of a complex
number. It is assumed that $B\ll f$.

The transmitted signal propagates through a wireless channel with
channel gain $h>0$ and phase shift $\theta\in[0,2\pi)$. The
equivalent complex channel is denoted by $\tilde{h}=\sqrt{h}e^{j\theta}$.
The noise $\Na$ after the
receiving antenna\footnote{The antenna noise may include
thermal noise from the transmitter and receiver chains.} can be modeled as a
narrow-band (with bandwidth $B$ and center frequency $f$) Gaussian
noise, i.e.,
$\Na=\sqrt{2}\Re \{ \Nax e^{j\W} \}$, where $\Nax=n_{\rm I}(t)+j n_{\rm Q}(t)$
with $n_{\rm I}(t)$ and $n_{\rm Q}(t)$
denoting the in-phase and quadrature noise components, respectively.
We assume that $n_{\rm I}(t)$ and $n_{\rm Q}(t)$ are independent Gaussian
random variables (RVs) with zero mean and variance $\sigma^2_{\rm A} /2$,
denoted by $\mathcal{N}(0,\sigma^2_{\rm A} /2)$, where $\sigma^2_{\rm A}=N_0 B$,
and $N_0$ is the one-sided noise power spectral density.
Thus, we have $\Nax \sim \mathcal{CN}(0,\sigma^2_{\rm A})$, i.e., $\Nax$
is a circularly symmetric complex Gaussian (CSCG) RV with zero mean
and variance $\sigma^2_{\rm A}$.
Corrupted by the antenna noise, the received signal $y(t)$ is given by
$y(t)=\sqrt{2}\Re \{\tilde y(t)\}$, where the complex signal $\tilde
y(t)$ is
\begin{align}\label{eq:yt-complex}
    \tilde y(t)=\Rp x(t) e^{j(\W+\theta)} + \Nax e^{j\W}.
\end{align}


\subsection{Information Receiver}

\begin{figure}
\begin{center}
\scalebox{0.48}{\includegraphics{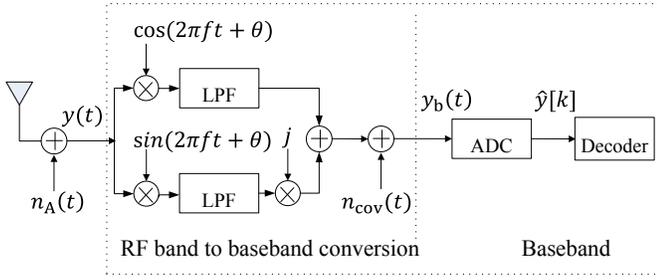}}
\end{center}
\caption{Information receiver.}
\label{fig:ID}
\end{figure}

First, we consider the case where the receiver shown in Fig.
\ref{fig:system model} is solely an information receiver.
Fig. \ref{fig:ID} shows the standard operations at an
information receiver with coherent demodulation (assuming that the
channel phase shift $\theta$ is perfectly known at the receiver).
The received RF band signal $y(t)$ is first converted to a complex
baseband signal $y_{\rm b}(t)$ and then sampled and digitalized
by an analog-to-digital converter (ADC) for further decoding.
The noise introduced by the RF band to baseband signal conversion
is denoted by $\Ncov$ with $\Ncov \sim \mathcal{CN}(0,\sigma^2_{\rm cov})$.
For simplicity, we assume an ideal ADC with zero
noise\footnote{The general case with nonzero ADC noise is considered in Remark \ref{remark:adc noise}.}.
The discrete-time ADC output is then given by
\begin{align}\label{eq:ID-yout}
    \Yo &=\Rp x[k]+ \tilde{n}_{\rm A}[k] + n_{\rm cov}[k]
\end{align}
where $k=1,2,\ldots$, denotes the symbol index.

It follows from (\ref{eq:ID-yout}) that the equivalent baseband
channel for wireless information transmission is the well-known
AWGN channel:
\begin{equation}\label{}
    Y=\Rp X+Z
\end{equation}
where $X$ and $Y$ denote the channel input and output, respectively,
and $Z\sim \mathcal{CN}(0,\sigma^2_{\rm A}+\sigma^2_{\rm cov})$ denotes the
complex Gaussian noise (assuming independent $\Nax$ and $\Ncov$).
When the channel input is distributed as $X\sim \mathcal{CN}(0,1)$,
the maximum achievable information rate (in bps/Hz) or the capacity
of the AWGN channel is given by \cite{Cover}
\begin{align}\label{eq:ID-rate}
      R&=\log_2\left(1+\frac{hP}{\sigma^2_{\rm A}+\sigma^2_{\rm cov}}\right).
\end{align}

\subsection{Energy Receiver} \label{subsec:energy receivers}

\begin{figure}
\begin{center}
\scalebox{0.48}{\includegraphics{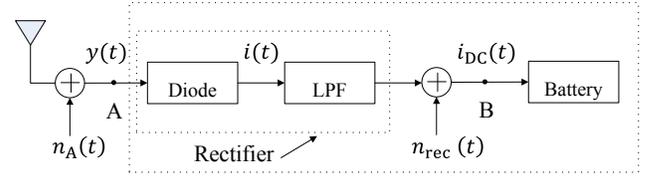}}
\end{center}
\caption{Energy receiver.}
\label{fig:EH}
\end{figure}

Next, we consider the case where the receiver in Fig.
\ref{fig:system model} is solely an energy receiver, and derive the
average wireless power that can be harvested from the received
signal. Fig. \ref{fig:EH} illustrates the operations of a typical
energy receiver that converts RF energy directly via a {\it rectenna}
architecture \cite{Paing}. In the rectenna, the received
RF band signal $y(t)$ is converted to a direct current (DC) signal
$i_{\rm DC}(t)$ by a rectifier, which consists of a Schottky diode and a
passive low-pass filter (LPF). The DC signal $i_{\rm DC}(t)$ is then used to
charge the battery to store the energy. With an input voltage
proportional to $y(t)$, the output current $i(t)$ of a Schottky
diode is given by \cite{Akkermans}:
\begin{equation}
\label{eq:it}
    i(t)=I_s \left( e^{\gamma y(t)}-1 \right)  =a_1 y(t)+a_2 y^2(t)+a_3 y^3(t)+ \cdots
\end{equation}
where $I_s$ denotes the saturation current, $\gamma$ denotes the
reciprocal of the thermal voltage of the Schottky diode, and the
coefficients $a_n$'s are given by $a_n=I_s \gamma^n /
n!,n=1,2,\ldots$, due to the Taylor series expansion of the exponential function.

From (\ref{eq:yt-complex}), for convenience we re-express $y(t)$ as
follows:
\begin{align}\label{eq:yt}
    y(t)
        &=\sqrt{2}\Re\{ \Rp x(t) e^{j(\W+\theta)} + \Nax e^{j\W}  \}  \nonumber \\
        &=\sqrt{2}\mu_{\rm Y}(t)\cos\left( \W +\phi_{\rm Y}(t) \right)
\end{align}
where $\phi_{\rm Y}(t)=\arctan \frac{\mu_{\rm Q}(t)}{\mu_{\rm I}(t)}$ and
\begin{align}\label{eq:muYt}
    \mu_{\rm Y}(t)=\sqrt{ \mu^2_{\rm I}(t) + \mu^2_{\rm Q}(t) }
\end{align}
with
\begin{align}\label{}
    & \mu_{\rm I}(t)=\Rp A(t)\cos\left( \phi(t)+\theta \right)  +n_{\rm I}(t)  \label{eq:muIt} \\
    & \mu_{\rm Q}(t)=\Rp A(t)\sin\left( \phi(t)+\theta \right)  +n_{\rm Q}(t). \label{eq:muQt}
\end{align}

By substituting (\ref{eq:yt}) into (\ref{eq:it}) and ignoring the
higher-order (larger than two) terms of $y(t)$, since $\gamma y(t)$ is practically a small number close to zero, we obtain
\begin{align}\label{}
    i(t)& \approx \sqrt{2}a_1 \mu_{\rm Y}(t)\cos\left(\W +\phi_{\rm Y}(t)\right) \nonumber \\
        & ~~~~ + 2a_2 \mu^2_{\rm Y}(t)\cos^2\left(\W +\phi_{\rm Y}(t)\right) \nonumber \\
        & = a_2 \mu^2_{\rm Y}(t) + \sqrt{2}a_1 \mu_{\rm Y}(t)\cos\left(\W +\phi_{\rm Y}(t)\right)  \nonumber \\
        & ~~~~ + a_2 \mu^2_{\rm Y}(t)\cos\left(4\pi ft +2\phi_{\rm Y}(t)\right).
\end{align}

The output current $i(t)$ of the diode is processed by a LPF,
through which the high-frequency harmonic components at both $f$ and $2f$ in $i(t)$
are removed and a DC signal $i_{\rm DC}(t)$ appears as the output of the
rectifier. Assuming that the additive noise introduced by the
rectifier is $n_{\rm rec}(t)$, the filtered output $i_{\rm DC}(t)$ is thus given by
\begin{align}\label{eq:idc-pre}
    i_{\rm DC}(t) = a_2 \mu^2_{\rm Y}(t) +n_{\rm rec}(t).
\end{align}
Since $a_2$ is a constant specified by the diode, for convenience we
assume in the sequel that $a_2=1$ (with $n_{\rm rec}(t)$
normalized accordingly to maintain the signal-to-noise ratio (SNR)).
Note that in (\ref{eq:idc-pre}), $a_2$ involves unit conversion from a power signal to a current
signal, thus by normalization $n_{\rm rec}(t)$ can be equivalently viewed as a power signal.
Assume $n_{\rm rec}(t)\sim \mathcal{N}(0,\sigma^2_{\rm rec})$, where $\Srec$ is in watt.
Substituting (\ref{eq:muYt}),
(\ref{eq:muIt}) and (\ref{eq:muQt}) into (\ref{eq:idc-pre}) yields
\begin{align}\label{eq:idc}
   & i_{\rm DC}(t)={\left( \Rp A(t)\cos\left( \phi(t)+\theta \right)  +n_{\rm I}(t)  \right)}^2   \nonumber \\
          &~~~~ + {\left( \Rp A(t)\sin\left( \phi(t)+\theta \right)  +n_{\rm Q}(t)  \right)}^2
              + n_{\rm rec}(t).
\end{align}

We assume that the converted energy to be stored in the battery is
linearly proportional to $i_{\rm DC}(t)$ \cite{Powercast}, with a
conversion efficiency $0<\zeta\leq1$. We also assume that the
harvested energy due to the noise (including both the antenna noise
and the rectifier noise) is a small constant and thus ignored.
Hence, the harvested energy (assuming the symbol period to be one) stored in
the battery, denoted by $Q$ in joule, is given by\footnote{For convenience,
in the sequel of the paper the two terms ``energy'' and
``power'' may be used interchangeably by assuming the symbol period to be one.}
\begin{align}\label{Q}
    Q=\zeta\mathbb{E}[i_{\rm DC}(t)]=\zeta hP.
\end{align}

\subsection{Performance Upper Bound}
Now consider the general case of interest where both information
decoding and energy harvesting are implemented at the receiver, as
shown in Fig. \ref{fig:system model}. Our main objective is to
maximize both the decoded information rate $R$ and harvested energy
$Q$ from the same received signal $y(t)$. Based on the results in
the previous two subsections, we derive an upper bound for the
performance of any practical receiver with the joint operation of
information decoding and energy harvesting, as follows. For
information transfer, according to the data-processing inequality
\cite{Cover}, with a given antenna noise $\Nax \sim \mathcal{
CN}(0,\sigma^2_{\rm A})$, the maximum information rate $R$ that can be
reliably decoded at the receiver is upper-bounded by $R\leq\log_2
(1+hP/\sigma^2_{\rm A})$. Note that state-of-the-art wireless information
receivers are not yet able to achieve this rate upper bound due to
additional processing noise such as the RF band to baseband
conversion noise $\Ncov$, as shown in (\ref{eq:ID-rate}). On the
other hand, for energy transfer, according to the law of energy
conservation, the maximum harvested energy $Q$ to be stored in the
battery cannot be larger than that received by the receiving
antenna, i.e., $Q\leq hP$. Note that practical energy receivers
cannot achieve this upper bound unless the energy conversion
efficiency $\zeta$ is made ideally equal to unity, as suggested by
(\ref{Q}). Following the definition of rate-energy (R-E) region
given in \cite{Varshney,Sahai,Zhang} to characterize
all the achievable rate (in bps/Hz for information transfer) and
energy (in joules/sec for energy transfer) pairs under a given
transmit power constraint $P$, we obtain a performance upper bound
on the achievable R-E region for the system in Fig. \ref{fig:system
model} as
\begin{equation}\label{eq:RE upper}
\hspace{-0.0cm} \mathcal{C}_{\rm R-E}^{\rm UB}(P)\triangleq\left\{(R,Q):
                   R\leq \log_2\left(1+\frac{hP}{\sigma^2_{\rm A}}\right), Q\leq hP \right\}
\hspace{-0.25cm}
\end{equation}
which is a box specified by the origin and the three vertices $(0,Q_{\rm max})$,
$(R_{\rm max}, 0)$ and $(R_{\rm max},Q_{\rm max})$, with $Q_{\rm
max} = hP$ and $R_{\rm max} =\log_2(1+hP/\sigma^2_{\rm A})$.
This performance bound is valid for all receiver architectures, some of which
will be studied next.

\section{Receiver Architecture for Wireless Information and Power Transfer}\label{sec:architecture}
This section considers practical receiver designs for
simultaneous wireless information and power transfer. We
propose a general receiver operation called {\it dynamic power
splitting} (DPS), from which we propose {\it separated} information
and energy receiver and {\it integrated} information and energy receiver.

\subsection{Dynamic Power Splitting}
Currently, practical circuits for harvesting energy from radio
signals are not yet able to decode the carried information directly.
In other words, the signal that is used for harvesting energy cannot
be reused for decoding information. Due to this potential
limitation, we propose a practical DPS scheme to enable the receiver
to harvest energy and decode information from the same received
signal at any time $t$, by dynamically splitting the signal into two
streams with the power ratio $\rho(t):1-\rho(t)$, which are used for harvesting
energy and decoding information, respectively, where $0\leq\rho(t)\leq1$.

Consider a block-based transmission of duration $T$ with $T=NT_{\rm s}$,
where $N$ denotes the number of transmitted symbols per block and
$T_{\rm s}$ denotes the symbol period.
We assume that $\rho(t)=\rho_k$ for
any symbol interval $t\in [(k-1)T_{\rm s},kT_{\rm s}),k=1,\ldots,N$. For
convenience, we define a power splitting vector as
${\mv{\rho}}=[\rho_1,\ldots,\rho_N]^T$.
In addition, in this paper
we assume an ideal power splitter \cite{Powerdivider,Agilent} at the receiver without any power
loss or noise introduced, and that the receiver can perfectly
synchronize its operations with the transmitter based on a given
vector ${\mv{\rho}}$.
During the transmission block time $T$, it is assumed that the information receiver
may operate in two modes: switch off (off mode) for a time duration $T_{\rm off}$
to save power, or switch on (on mode) for a time duration $T_{\rm on}=T-T_{\rm off}$
to decode information. The percentage of time that the information decoder operates
in off mode is denoted by $\alpha$ with $0\leq\alpha\leq1$,
thus we have $T_{\rm off}=\alpha T$ and $T_{\rm on}=(1-\alpha)T$.
Without loss of generality, we assume that the information receiver operates in off
mode during the first $\lfloor\alpha N\rfloor$ symbols during each block with
$k=1,\ldots,\lfloor\alpha N\rfloor$, where $\lfloor\cdot\rfloor$ denotes the floor
operation, while in on mode during the remaining symbols with
$k=\lfloor\alpha N\rfloor +1,\ldots,N$.
For convenience, we also assume in the sequel that
$\alpha N$ is a positive integer regardless of the value of
$\alpha$, which is approximately true if $N$ is chosen to be a very
large number in practice.

Next, we investigate three special cases of DPS, namely {\it time switching} (TS),
{\it static power splitting} (SPS) and {\it on-off power splitting} (OPS) given in \cite{Zhang}:
\begin{itemize}
\item {\it Time switching (TS)}: With TS, for the first $\alpha N$ symbols
    when the information receiver operates in off mode, all signal power is used for
    energy harvesting. For the remaining
    $(1-\alpha)N$ symbols when the information receiver operates
    in on mode, all signal power is used for information decoding.
    Thus for TS, we have
    \begin{equation}\label{eq:rho-TS}
    \rho_k=\begin{cases}
           1, & k=1,\ldots,\alpha N  \\
           0, & k=\alpha N+1,\ldots,N.
           \end{cases}
    \end{equation}

\item {\it Static power splitting (SPS)}: With SPS, the information receiver operates in
    on mode for all $N$ symbols, i.e., $\alpha=0$. Moreover, the ratio of the split signal power
    for harvesting energy and decoding information is set to be a constant $\rho$ for all $N$ symbols.
    Thus for SPS, we have
    \begin{equation}\label{eq:rho-SPS}
        \rho_k=\rho, \ k=1,\ldots,N.
    \end{equation}

\item{\it On-off power splitting (OPS)}: With OPS, for the first $\alpha N$ symbols
    all signal power is used for energy harvesting. For the remaining
    $(1-\alpha)N$ symbols, the ratio of the split signal power
    for harvesting energy and decoding information is set to be a constant $\rho$, with $0\leq\rho<1$.
    Thus, for a given power splitting pair $(\alpha, \rho)$, we have
    \begin{align}\label{eq:rho-OPS}
        \rho_k=\begin{cases}
           1, & k=1,\ldots,\alpha N  \\
           \rho, & k=\alpha N+1,\ldots,N.
           \end{cases}
    \end{align}
    Note that TS and SPS are two special cases of OPS by letting $\rho=0$ (for TS)
    or $\alpha=0$ (for SPS) in (\ref{eq:rho-OPS}).
\end{itemize}

\subsection{Separated vs. Integrated Receivers}

\begin{figure}
\begin{center}
\scalebox{0.48}{\includegraphics{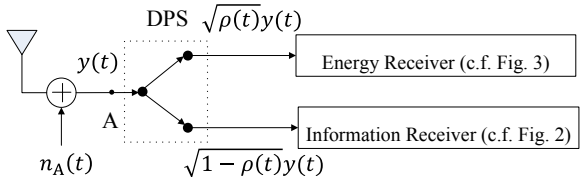}}
\end{center}
\caption{Architecture for the separated information and energy receiver.}
\label{fig:SepRx}
\end{figure}

\begin{figure}
\begin{center}
\scalebox{0.48}{\includegraphics{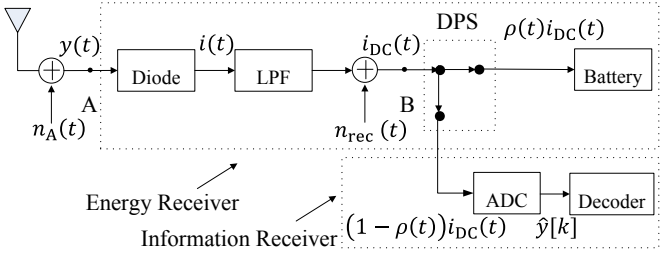}}
\end{center}
\caption{Architecture for the integrated information and energy receiver.}
\label{fig:IntRx}
\end{figure}

In this subsection, we propose two types of receivers that exploit
the DPS scheme in different ways. The first type of receivers is called
{\it separated} information and energy receiver, as shown in Fig.
\ref{fig:SepRx}, while the second type is called {\it integrated}
information and energy receiver, as shown in Fig. \ref{fig:IntRx}.
These two types of receivers both use the energy receiver
in Fig. \ref{fig:EH} for energy harvesting. Their
difference lies in that for the case of separated receiver, the
power splitter for DPS is inserted at point `A' in the RF band of
the energy receiver shown in Fig. \ref{fig:EH}, while in the
case of integrated receiver, the power splitter is inserted at
point `B' in the baseband.

First, we consider the case of separated information and energy
receiver. As shown in Fig. \ref{fig:SepRx}, a power splitter is
inserted at point `A', such that the received signal $y(t)$ by the
antenna is split into two signal streams with power levels specified
by $\rho(t)$ in the RF band, which are then separately fed to the
conventional energy receiver (cf. Fig. \ref{fig:EH}) and information
receiver (cf. Fig. \ref{fig:ID}) for harvesting energy and decoding
information, respectively. The achievable R-E region for this type
of receivers with DPS will be studied in Section \ref{sec:separated
receivers}.

Next, we consider the integrated information and energy receiver,
as motivated by the following key observation. Since the transmitted
power in a wireless power transfer system can be varied over time
provided that the average power delivered to the receiver is above a
certain required target, we can encode information in the energy
signal by varying its power levels over time, thus achieving
continuous information transfer without degrading the power transfer
efficiency.
To emphasize this dual use of signal power in both WPT as well as WIT, the
modulation scheme is called {\em energy modulation}. A constellation example,
namely, {\em pulse energy modulation}, is provided later in
Section \ref{sec:practical modulation}.
Note that to decode the energy
modulated information at the receiver, we need to detect the power
variation in the received signal within a certain accuracy, by
applying techniques such as {\em energy detection} \cite{Urkowitz}.
Recall that in Section \ref{subsec:energy receivers}, for the energy
receiver in Fig. \ref{fig:EH}, the received RF signal $y(t)$ is
converted to a DC signal $i_{\rm DC}(t)$ given in (\ref{eq:idc}) by a
rectifier. Note that this RF to DC conversion is analogous to the
RF band to baseband conversion in conventional wireless information
receivers in Fig. \ref{fig:ID}. Thus, $i_{\rm DC}(t)$ can be treated as
a baseband signal for information decoding (via energy detection).

Based on the above observation, we propose the integrated
information and energy receiver as shown in Fig. \ref{fig:IntRx},
by inserting a power splitter at point `B' of the conventional
energy receiver. With DPS, $i_{\rm DC}(t)$ is split into two portions
specified by $\rho(t)$ for energy harvesting and information
decoding, respectively. Note that unlike the traditional information
receiver in Fig. \ref{fig:ID}, the information receiver in the case
of integrated receiver does not implement any RF band to baseband
conversion, since this operation has been {\em integrated} to the
energy receiver (via the rectifier). The achievable R-E region for
this type of receivers will be studied in Section
\ref{sec:integrated receivers}.

\section{Rate-Energy Tradeoff for Separated Information and Energy Receiver}\label{sec:separated receivers}
In this section, we study the achievable R-E region for the
separated information and energy receiver shown in Fig.
\ref{fig:SepRx}.
With DPS, the average SNR
at the information receiver for the $k$-th transmitted symbol,
$k=1,\ldots,N$, is denoted by $\tau(\rho_k)$, and given by
\begin{equation}\label{eq:SNR}
    \tau(\rho_k)=\frac{(1-\rho_k)hP}{(1-\rho_k)\sigma^2_{\rm A}+\sigma^2_{\rm cov}}.
\end{equation}
From (\ref{eq:SNR}), we obtain the achievable R-E region for the DPS
scheme in the case of separated receiver as
\begin{align}\label{eq:RE DPS}
    \mathcal{C}_{\rm R-E}^{\rm DPS}
     &  (P)\triangleq \bigcup\limits_{\mv{\rho}}\left\{(R,Q):
      Q\leq\frac{1}{N}\sum\limits_{k=1}^{N}\rho_k \zeta hP,\right.\nonumber \\
     & \left. R\leq \frac{1}{N}\sum\limits_{k=1}^{N} \log_2\left(1+\frac{(1-\rho_k)hP}{(1-\rho_k)\sigma^2_{\rm A}+\sigma^2_{\rm cov}} \right) \right\}.
\end{align}

Next, we address the two special cases of DPS, i.e., the TS scheme
and the SPS scheme. Substituting (\ref{eq:rho-TS}) into (\ref{eq:RE
DPS}), the achievable R-E region for the TS scheme is given by
\begin{align}\label{eq:RE TS}
    \mathcal{C}_{\rm R-E}^{\rm TS}(P)\triangleq
    &  \bigcup\limits_{\alpha}\bigg\{ (R,Q): Q\leq \alpha \zeta hP, \nonumber \\
    &  \left.  R\leq  (1-\alpha)\log_2 \left(1+\frac{hP}{\sigma^2_{\rm A}+\sigma^2_{\rm cov}}\right)  \right\}.
\end{align}
Let $\hat R_{\rm max}=\log_2\left(1+hP/(\sigma^2_{\rm A}+\sigma^2_{\rm cov})\right)$ given in
(\ref{eq:ID-rate}) and $\hat Q_{\rm max}=\zeta hP$ given in (\ref{Q}). It
is noted that the boundary of $\mathcal{C}_{\rm R-E}^{\rm TS}(P)$ is
simply a straight line connecting the two points $(\hat R_{\rm max},0)$
and $(0,\hat Q_{\rm max})$ as $\alpha$ sweeps from 0 to 1.

Substituting (\ref{eq:rho-SPS}) into (\ref{eq:RE DPS}), the achievable R-E region for the SPS scheme is given by
\begin{align}\label{eq:RE SPS}
    \mathcal{C}_{\rm R-E}^{\rm SPS}(P)\triangleq
    &  \bigcup\limits_{\rho}\bigg\{ (R,Q): Q\leq \rho \zeta hP,  \nonumber \\
    &  \left. R\leq\log_2 \left(1+\frac{(1-\rho)hP}{(1-\rho)\sigma^2_{\rm A}+\sigma^2_{\rm cov}}\right)
             \right\}.
\end{align}

\begin{proposition}\label{proposition:1}
For the separated information and energy receiver, the SPS scheme
is the optimal DPS scheme, i.e., $\mathcal{C}_{\rm R-E}^{\rm
DPS}(P)=\mathcal{C}_{\rm R-E}^{\rm SPS}(P),P\geq 0$.
\end{proposition}
\begin{proof}
Please refer to Appendix \ref{appendix:proof prop1}.
\end{proof}

From Proposition \ref{proposition:1}, it suffices for us to consider
the SPS scheme for the optimal R-E tradeoff in the case of separated
receivers. In particular, if $\sigma^2_{\rm A} \ll \sigma^2_{\rm cov}$, i.e.,
the processing noise is dominant over the antenna noise, from
(\ref{eq:SNR}) the SNR at the information receiver
$\tau(\rho)\rightarrow (1-\rho)hP/\sigma^2_{\rm cov}$.
In the other extreme case with
$\sigma^2_{\rm A} \gg \sigma^2_{\rm cov}$, from (\ref{eq:SNR}) we have
$\tau(\rho)\rightarrow hP/\sigma^2_{\rm A}$, which is independent of
$\rho$. Thus, the optimal rate-energy tradeoff is achieved when
infinitesimally small power is split to the information receiver,
i.e., $\rho\rightarrow 1$. In this case, it can be shown that when $\zeta=1$,
$\mathcal{C}_{\rm R-E}^{\rm SPS}(P)\rightarrow \mathcal{C}_{\rm
R-E}^{\rm UB}(P)$, which is the R-E tradeoff outer bound given in
(\ref{eq:RE upper}).

\begin{figure}
\centering
 \epsfxsize=0.7\linewidth
    \includegraphics[width=9cm]{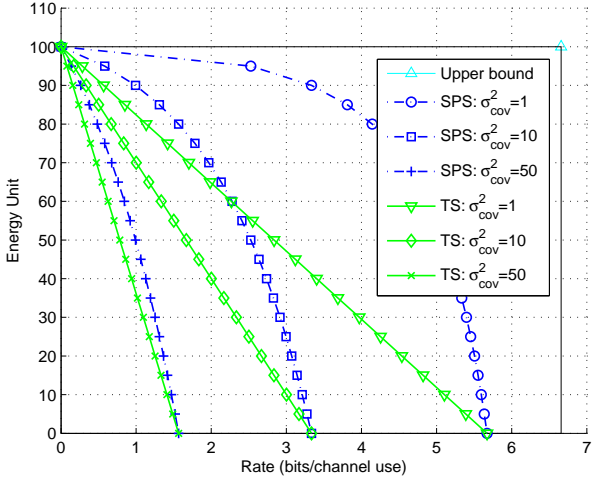}
\caption{Rate-energy tradeoff for TS vs. SPS based separated receiver with $h=1, P=100, \zeta=1$ and $\sigma^2_{\rm A}=1$.}
\label{fig:RE-Sep}
\end{figure}

Fig. \ref{fig:RE-Sep} shows the achievable
R-E regions under different noise power setups for the
separated information and energy receiver (SepRx).
It is assumed that $h=1$, $P=100$, $\zeta=1$, and the antenna noise
power is set to be $\Va=1$.
With normalization, for convenience
we denote the information rate and harvested energy in terms of
bits/channel use and energy unit, respectively. In Fig. \ref{fig:RE-Sep}, it is
observed that for SepRx, the SPS scheme always achieves larger R-E
pairs than the TS scheme for different values of the processing
(RF band to baseband conversion) noise power $\sigma^2_{\rm cov}$.
Moreover, as $\sigma^2_{\rm cov}$ increases, the gap between
$\mathcal{C}_{\rm R-E}^{\rm TS}(P)$ and $\mathcal{C}_{\rm R-E}^{\rm
SPS}(P)$ shrinks, while as $\sigma^2_{\rm cov}$ decreases, the
achievable R-E region with SPS enlarges and will eventually approach to
the R-E region upper bound given in (\ref{eq:RE upper}) when
$\sigma^2_{\rm cov}\rightarrow 0$.

\section{Rate-Energy Tradeoff for Integrated Information and Energy Receiver} \label{sec:integrated receivers}

In this section, we study the rate-energy performance for the
integrated information and energy receiver shown in Fig.
\ref{fig:IntRx}.
In the integrated receiver, due to the RF to baseband conversion by the rectifier,
we shall see that the equivalent baseband channel is nonlinear,
as opposed to that of the separated receiver where the channel is linear.

From (\ref{eq:idc}), for convenience we re-express
$i_{\rm DC}(t)$ as follows:
\begin{align}\label{eq:idc-simp-pre}
    i_{\rm DC}(t)= \left| \sqrt{hP} A(t) e^{j\left(\theta+\phi(t)\right)} +\Nax  \right|^2   + n_{\rm rec}(t).
\end{align}
Since planar rotation does not change the statistics of $\Nax$, (\ref{eq:idc-simp-pre}) can be equivalently written as
\begin{align}\label{eq:idc-simp}
    i_{\rm DC}(t)= \left| \sqrt{hP} A(t) +\Nax  \right|^2   + n_{\rm rec}(t).
\end{align}

As shown in Fig. \ref{fig:IntRx}, after the noiseless power splitter
and ADC, the output $\Yonew,k=1,\ldots,N$, is given by
\begin{align}\label{eq:Int-yout}
    \Yonew = \left(1-\rho_k\right) \left(\left| \sqrt{hP} A[k]  +\tilde n_{\rm A}[k]\right|^2 + n_{\rm rec}[k] \right).
\end{align}
In the above it is worth noting that the average SNR
at any $k$ is independent of $\rho_k$ provided that $\rho_k<1$.
Thus, to minimize the power split for information decoding (or
maximize the power split for energy harvesting), we should let
$\rho_k\rightarrow 1, \forall k$, i.e., splitting infinitesimally
small power to the information receiver all the time. Thereby, DPS
becomes an equivalent SPS with $\rho\rightarrow 1$ in the case of
integrated receiver.

With $\rho_k$'s all equal to 1 in (\ref{eq:Int-yout}), the
equivalent discrete-time memoryless channel for the information decoder is modeled as
\begin{align}\label{eq:new model}
    Y = \left| \sqrt{hPX} +Z_2 \right|^2 + Z_1
\end{align}
where $X$ denotes the signal power, which is the nonnegative channel input; $Y$ denotes the channel output;
$Z_2\sim \mathcal{CN}(0,\sigma^2_{\rm A})$ denotes the antenna noise; and
$Z_1\sim \mathcal{N}(0,\sigma^2_{\rm rec})$ denotes the rectifier noise.
It is worth noting that for the channel (\ref{eq:new model}) information is encoded
in the power (amplitude) of the transmitted signal $x(t)$, rather than the phase of $x(t)$.
The channel in (\ref{eq:new model}) is nonlinear and thus it is
challenging to determine its capacity $C_{\rm NL}$ and corresponding optimal input distribution subject to
$X\geq0$ and $\mathbb{E}[X]\leq1$, where $X$ is real.
Similar to the case of separated receiver, we consider the following two
special noise power setups:
\begin{itemize}
    \item Case 1 (Negligible Antenna Noise) with $\Va\rightarrow 0$:
    In practice, this case may be applicable when the antenna noise power is much smaller than the rectifier noise power,
    thus the antenna noise can be omitted.
    With $\Va\rightarrow 0$, we have $Z_2\rightarrow 0$. Thus, the channel in (\ref{eq:new model}) becomes
        \begin{equation}\label{eq:sim model-Na=0}
            Y=hP X+ Z_1
        \end{equation}
        where $X\geq0$ and real-valued, which is known as the {\em optical intensity channel}. It is shown in \cite{Faycal}
        that the optimal input distribution to this channel is discrete.
        According to \cite{Lapidoth-Moser}, the capacity $C_1$ for the channel (\ref{eq:sim model-Na=0})
        is upper-bounded by

        \begin{small}\vspace{-0.1in}
        \begin{align}\label{eq:C1-ub}
           &  C_1^{\rm ub} = \log_2\left(\beta e^{-\frac{\delta^2}{2\sigma^2_{\rm rec}}}+\sqrt{2\pi}\sigma_{\rm rec}\mathcal{Q}\left(\frac{\delta}{\sigma_{\rm rec}}\right)\right)       \nonumber\\
             &  + \left( \frac{1}{2}\mathcal{Q}\left(\frac{\delta}{\sigma_{\rm rec}}\right)
                + \frac{1}{\beta}\left(\delta+hP+\frac{\sigma_{\rm rec}e^{-\frac{\delta^2}{2\sigma^2_{\rm rec}}}}{\sqrt{2\pi}}\right) \right) \log_2 e     \nonumber\\
             &  + \left(\frac{\delta e^{-\frac{\delta^2}{2\sigma^2_{\rm rec}}}}{2\sqrt{2\pi}\sigma_{\rm rec}}
                + \frac{\delta^2}{2\sigma^2_{\rm rec}} \left( 1-\mathcal{Q}\left(\frac{\delta+hP}{\sigma_{\rm rec}}\right)\right)\right) \log_2 e \nonumber\\
             &  - \frac{1}{2}\log_2 2\pi e \sigma^2_{\rm rec}
        \end{align}
        \end{small}where $\mathcal{Q}(\cdot)=\frac{1}{\sqrt{2\pi}}\int_x^{\infty}e^{-\frac{t^2}{2}}d t$ denotes the $\mathcal{Q}$-function, and $\beta>0$, $\delta\geq0$ are free parameters. The details of choice for $\beta$ and $\delta$ are provided in \cite{Lapidoth-Moser}, and thus are omitted in this paper for brevity.
        Moreover, the asymptotic capacity at high power ($P\rightarrow \infty$) is given by \cite{Lapidoth-Moser}
        \begin{equation}\label{eq:C1-inf}
            C_1^{\infty} = \log_2 \frac{hP}{\sigma_{\rm rec}} + \frac{1}{2}\log_2 \frac{e}{2\pi}.
        \end{equation}

    \item Case 2 (Negligible Rectifier Noise) with $\Srec \rightarrow 0$:
    This case is applicable when the antenna noise power is much greater than the rectifier noise power;
    thus, the rectifier noise can be omitted. With $\Srec\rightarrow 0$, we have $Z_1\rightarrow 0$.
    The channel in (\ref{eq:new model}) is then simplified as
        \begin{equation}\label{eq:sim model-Np=0}
            Y=  \left| \sqrt{hPX} +Z_2 \right|^2
        \end{equation}
        which is equivalent to the {\em noncoherent AWGN channel}. It is shown in \cite{Shamai} that the optimal input distribution to this channel is discrete and possesses an infinite number of mass points.
        The capacity $C_2$ for the channel (\ref{eq:sim model-Np=0}) is upper-bounded by \cite{Shamai}
        \begin{equation}\label{eq:C2-ub}
            C_2^{\rm ub} = \frac{1}{2}\log_2\left( 1+ \frac{hP}{\sigma^2_{\rm A}} \right)
            +\frac{1}{2}\left( \log_2\frac{2\pi}{e}-C_E\log_2 e \right)
        \end{equation}
        where $C_E$ is Euler's constant.
        Moreover, the asymptotic capacity at high power ($P\rightarrow \infty$) is given by \cite{Shamai,Lapidoth}
        \begin{equation}\label{eq:C2-inf}
            C_2^{\infty} = \frac{1}{2}\log_2 \left( 1+\frac{hP}{2\sigma^2_{\rm A}} \right)
        \end{equation}
        which is achieved by choosing $X$ as central chi-square distribution with one degree of freedom\footnote{In this case, the input amplitude is distributed as the positive normal distribution, with probability density function $f_{\rm A}(a)=\sqrt{\frac{2}{\pi}}e^{-\frac{a^2}{2}}$.}.
\end{itemize}

In general, the capacity $C_{\rm NL}$ of the channel given in (\ref{eq:new model})
can be upper-bounded by
\begin{equation}\label{}
    C_{\rm NL}^{\rm ub} = \min\{C_1^{\rm ub}, C_2^{\rm ub}\}
\end{equation}
and the capacity lower bound $C_{\rm NL}^{\rm lb}$ for the channel (\ref{eq:new model})
can be computed by the mutual information obtained from any input distribution
satisfying the constraint $X\geq0$ and $\mathbb{E}[X]\leq1$.
It is worth noting that at high power ($P\rightarrow\infty$) from (\ref{eq:C1-inf}) and (\ref{eq:C2-inf}),
$C_1^{\infty}$ grows like $\log_2 P$; while $C_2^{\infty}$ grows like $\frac{1}{2}\log_2 P$.
Thus the channel (\ref{eq:sim model-Np=0}) provides a tighter upper bound for the asymptotic capacity of the
channel (\ref{eq:new model}) than the channel (\ref{eq:sim model-Na=0}) at high SNR.

\begin{figure}
\centering
 \epsfxsize=0.7\linewidth
    \includegraphics[width=9cm]{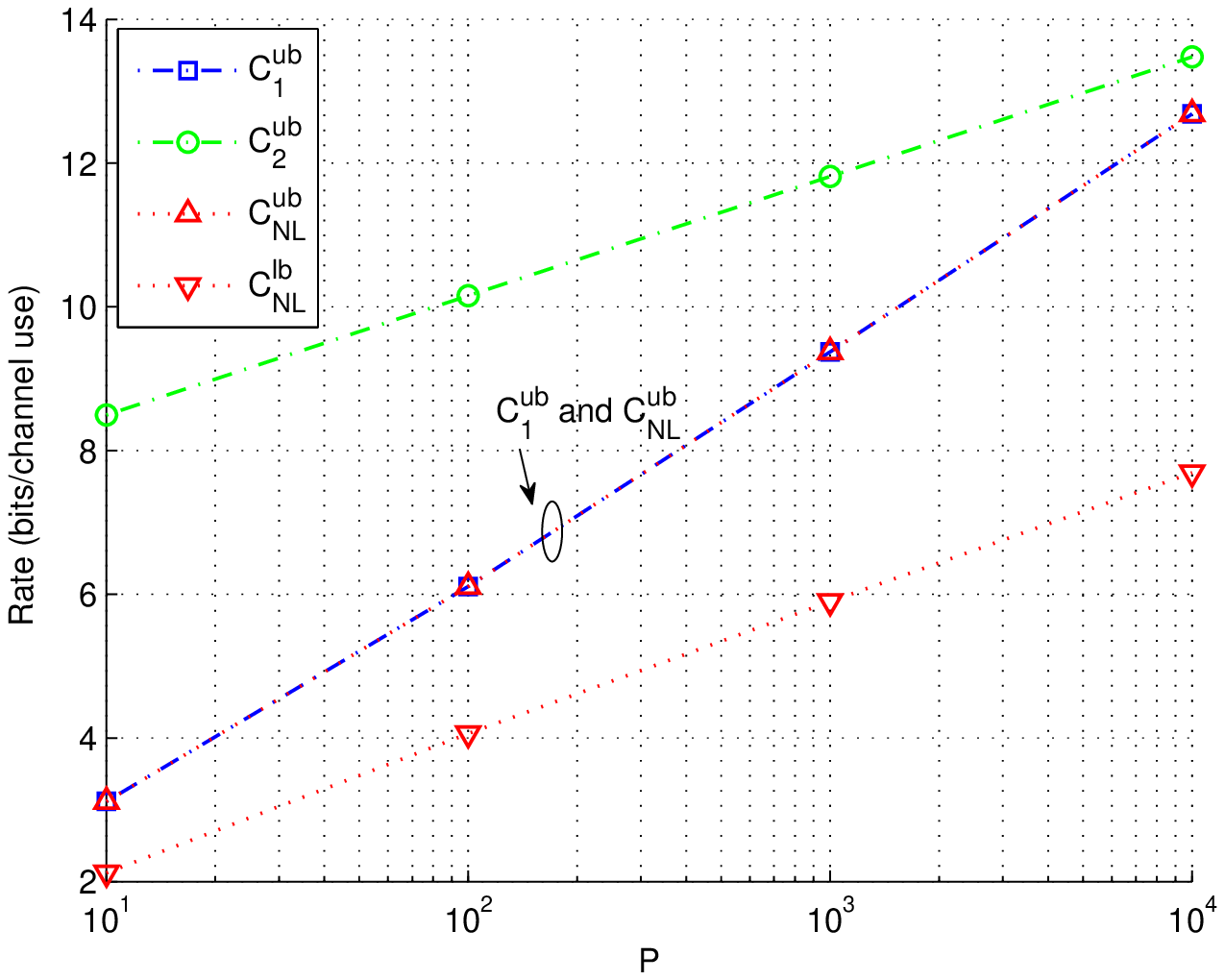}
\caption{Capacity bounds for the channels (\ref{eq:new model}), (\ref{eq:sim model-Na=0}) and (\ref{eq:sim model-Np=0}) with $h=1, \Va=10^{-4}$ and $\Srec=1$.}
\label{fig:Capacity}
\end{figure}

Fig. \ref{fig:Capacity} shows the capacity bounds for the above three
channels (\ref{eq:new model}), (\ref{eq:sim model-Na=0}) and (\ref{eq:sim model-Np=0}).
It is assumed that $h=1, \Va=10^{-4}$ and $\Srec=1$. The capacity lower bound $C_{\rm NL}^{\rm lb}$
for the channel given in (\ref{eq:new model}) is computed by
assuming the input (power) distribution is a central chi-square distribution with one degree of freedom.
We shall use this lower bound as the achievable rate for the integrated receiver in  the subsequent numerical results.
It is observed that in this case with dominant rectifier noise, the capacity upper bound $C_1^{\rm ub}$ in (\ref{eq:C1-ub}) is tighter than $C_2^{\rm ub}$ in (\ref{eq:C2-ub}). It is also observed that the gap between the capacity upper and lower bounds, namely $C_{\rm NL}^{\rm ub}$ and $C_{\rm NL}^{\rm lb}$, is still notably large under this setup, which can be further reduced by optimizing the input distribution.

To summarize, the achievable R-E region for the case of integrated
receivers by SPS with $\rho\rightarrow 1$ is given by
\begin{align}\label{eq:intRx R-E}
    \mathcal{C}_{\rm R-E}^{\rm SPS}(P)\triangleq\left\{(R,Q): R\leq
            C_{\rm NL}(P), Q\leq \zeta hP
            \right\}
\end{align}
where $C_{\rm NL}(P)$ denotes the capacity of the nonlinear (NL)
channel given in (\ref{eq:new model}) subject to
$X\geq0$ and $\mathbb{E}[X]\leq1$.

\begin{remark}\label{remark:adc noise}
We have characterized the rate-energy performance for the integrated receiver assuming
an ideal ADC with zero quantization noise. Now we extend our results to the case of nonzero quantization noise $n_{\rm ADC}(t)$.
It is assumed that $n_{\rm ADC}(t)\sim \mathcal{N}(0,\sigma^2_{\rm ADC})$ for
the integrated receiver \cite{Carbone,Ruscak}.
With nonzero ADC noise, (\ref{eq:Int-yout}) is modified as
\begin{align}\label{eq:Int-yout-adc}
    \Yonew = & \left(1-\rho_k\right) \left(\left| \sqrt{hP} A[k]  +\tilde n_{\rm A}[k]\right|^2 + n_{\rm rec}[k] \right) \nonumber\\
    &  + n_{\rm ADC}[k].
\end{align}
Thus, for given $k$ the equivalent channel in (\ref{eq:new model}) still holds, where $Z_1\sim\mathcal{N}\left(0,\sigma^2_{\rm rec}+\frac{\sigma^2_{\rm ADC}}{(1-\rho_k)^2}\right)$ denotes the equivalent processing noise. It is worth noting that
the equivalent processing noise power is a function of the power splitting ratio $\rho_k$; thus, the capacity in channel (\ref{eq:new model}) is also a function of $\rho_k$.
The achievable R-E region for the integrated receiver by DPS is thus given by
\begin{align}\label{eq:RE-Int-adc}
    \mathcal{C}_{\rm R-E}^{\rm DPS}(P)\triangleq\bigcup\limits_{\mv{\rho}}\bigg\{
    	& (R,Q): R\leq\frac{1}{N}\sum\limits_{k=1}^{N}C_{\rm NL}(P,\rho_k), 	\nonumber\\
           &  \left. Q\leq \frac{1}{N}\sum\limits_{k=1}^{N}\rho_k\zeta hP \right\}.
\end{align}
For the separated receiver, the results in Section \ref{sec:separated receivers}
can be easily extended to the case with nonzero ADC noise by adding the ADC
noise power to the total processing noise power.
\end{remark}

Figs. \ref{fig:RE-Int} and \ref{fig:RE-Int-adc} show the achievable
R-E regions under different noise power setups for both cases of
SepRx and IntRx. For both figures,
it is assumed that $h=1, P=100, \zeta=0.6$, and $\sigma^2_{\rm A}=1$.
In Fig. \ref{fig:RE-Int}, it is assumed that $\sigma^2_{\rm ADC}=0$.
In Fig. \ref{fig:RE-Int-adc}, it is assumed that $\sigma^2_{\rm ADC}=1$,
and $\rho_k=\rho, \forall k$ in (\ref{eq:RE-Int-adc}) with $0\leq \rho\leq1$.
Note that in practice, the degradation of ADC noise is usually
modeled by a so-called signal-to-quantization-noise ratio (SQNR), approximately
given by $6K$ dB, where $K$ is the number of quantization bits. Here,
by assuming $P=100$ and $\sigma^2_{\rm ADC}=1$, the SQNR equals to 20dB,
which implies $K \approx 3.3$bits. It follows that the number of quantization
levels is approximately 10.
In Figs. \ref{fig:RE-Int} and \ref{fig:RE-Int-adc},
the achievable rates for
IntRx are computed as the capacity lower bound for the channel given
in (\ref{eq:new model}) assuming the input as central chi-square
distribution with one degree of freedom.

As shown in Fig. \ref{fig:RE-Int}, the achievable R-E regions
for IntRx with zero ADC noise are marked by boxes as given in (\ref{eq:intRx R-E}).
In addition, when the processing
noise power ($\Vcov$ for SepRx and $\Srec$ for IntRx)
equals to the antenna noise power, i.e.,
$\Va=\Vcov=\Srec=1$, the achievable rate for IntRx is
notably lower than that for SepRx, due to the use of noncoherent
(energy) modulation by IntRx as compared to the use of coherent
modulation by SepRx. However, when the processing noise power is much greater
than the antenna noise power (as in most practical systems),
the achievable R-E region of IntRx becomes superior compared to that of SepRx with the same
processing noise power, i.e., $\Vcov=\Srec=100$.
This is due to the fact that for IntRx, the
processing (rectifier) noise incurs prior to the power splitter and
thus only infinitesimally small power is required to be split by the
power splitter to implement the energy detection for information
decoding (cf. (\ref{eq:new model})), while for SepRx, more
power needs to be split to the information decoder to
compensate for the processing (RF band to baseband conversion) noise
that incurs after the power splitter.
Moreover, in Fig. \ref{fig:RE-Int} it is observed that IntRx is more
suitable than SepRx when more wireless power is desired.

In Fig. \ref{fig:RE-Int-adc}, it is observed that the achievable R-E regions
for IntRx with nonzero ADC noise are no longer boxes.
Comparing Fig. \ref{fig:RE-Int-adc} with Fig. \ref{fig:RE-Int}, it is observed
that the achievable rate by IntRx with nonzero ADC noise is less than that by
IntRx with zero ADC noise, especially when more harvested energy is desired.

\begin{figure}
\centering
 \epsfxsize=0.7\linewidth
    \includegraphics[width=9cm]{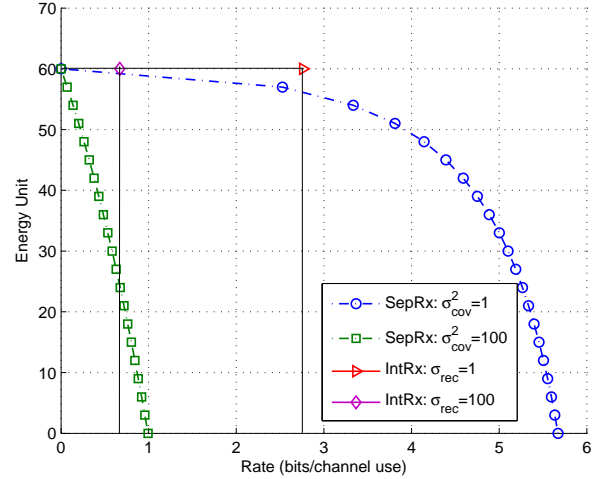}
\caption{Rate-energy tradeoff for separated vs. integrated receivers with $h=1, P=100, \zeta=0.6$, $\sigma^2_{\rm A}=1$ and $\sigma^2_{\rm ADC}=0$.}
\label{fig:RE-Int}
\end{figure}

\begin{figure}
\centering
 \epsfxsize=0.7\linewidth
    \includegraphics[width=9cm]{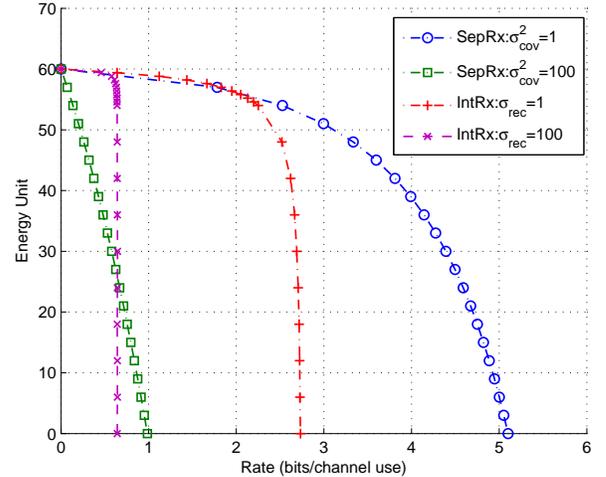}
\caption{Rate-energy tradeoff for separated vs. integrated receivers with $h=1, P=100, \zeta=0.6$ and $\sigma^2_{\rm A}=\sigma^2_{\rm ADC}=1$.}
\label{fig:RE-Int-adc}
\end{figure}

\section{Rate-Energy Tradeoff with Receiver Circuit Power Consumption}\label{sec:circuit power}
In Sections \ref{sec:separated receivers} and \ref{sec:integrated receivers},
the harvested energy is characterized as the energy harvested by the energy receiver
without consideration of power consumption by the receiver circuits.
For energy receiver, there is no energy consumption since both the Schottky diode and LPF are
passive devices\footnote{In practice, some RF energy harvesting systems have additional control circuits
which consume power, however, this power consumption has been included in the
conversion efficiency $\zeta$.}.
However, for information receiver, some amount of power will be consumed to supply the
information decoding circuits. In particular, for the separated receiver shown in Fig. \ref{fig:SepRx},
the circuit power consumed by information decoding, denoted by $\Pcs$, is given by
$\Pcs=P_{\rm m}+P_{\rm ADC}$, where $P_{\rm m}$ and $P_{\rm ADC}$ denote the power consumed by the RF band mixer and the ADC, respectively. For the integrated receiver shown in Fig. \ref{fig:IntRx}, however,
the circuit power consumed by information decoding, denoted by $\Pci$, is only given by
$\Pci=P_{\rm ADC}$.\footnote{Here $\Pcs$ and $\Pci$ are defined according to the proposed architectures
in Fig. \ref{fig:SepRx} and Fig. \ref{fig:IntRx}, respectively. In practice, the information decoding circuits may contain
additional components, such as a low noise amplifier (LNA) in the separated receiver. In general, the power consumed by the additional components can be added to $\Pcs$ or $\Pci$.}
Note that in general $\Pcs$ will be much greater than $\Pci$, since the RF band mixer consumes
comparable amount of power as compared to the ADC.
Thus the {\it net energy} stored in the battery will
be the harvested energy subtracted by that consumed by information decoding circuits.
In this section, we study the rate-energy tradeoff for both separated and integrated
receivers with receiver circuit power consumption taken into account.

\subsection{Separated Receiver with $\Pcs>0$}
For the separated receiver shown in Fig. \ref{fig:SepRx},
by modifying (\ref{eq:RE DPS}) to account for the circuit power $\Pcs$,
the achievable R-E region for the DPS
scheme is given by

\begin{footnotesize}\vspace{-0.2in}
\begin{align}\label{eq:RE DPS-Pmix}
    \mathcal{C}_{\rm R-E}^{\rm DPS'} &
       (P)\triangleq \bigcup\limits_{\mv{\rho}}\left\{
      (R,Q): 0\leq Q\leq\frac{1}{N}\left(\sum\limits_{k=1}^{N}\rho_k \zeta hP -\sum\limits_{k=\alpha N+1}^{N}\Pcs\right),\right.\nonumber \\
     & \left. R\leq \frac{1}{N}\sum\limits_{k=\alpha N+1}^{N} \log_2\left(1+\frac{(1-\rho_k)hP}{(1-\rho_k)\sigma^2_{\rm A}+\sigma^2_{\rm cov}} \right) \right\}.
\end{align}
\end{footnotesize}Next, we address one special case of DPS, i.e., the OPS scheme.
Substituting (\ref{eq:rho-OPS}) into (\ref{eq:RE DPS-Pmix}), the achievable R-E region for the OPS scheme is given by
\begin{align}\label{eq:RE OPS-Pmix}
  &  \mathcal{C}_{\rm R-E}^{\rm OPS'}
       (P)\triangleq \bigcup\limits_{\alpha,\rho}\bigg\{
        (R,Q): 0\leq Q\leq \alpha\zeta hP + (1-\alpha)\rho\zeta hP  \nonumber \\
       & \left. -(1-\alpha)\Pcs, R\leq (1-\alpha)\log_2\left(1+\frac{(1-\rho)hP}{(1-\rho)\sigma^2_{\rm A}+\sigma^2_{\rm cov}} \right) \right\}.
\end{align}

\begin{proposition}\label{proposition:2}
For the separated information and energy receiver with $\Pcs>0$, the OPS scheme
is the optimal DPS scheme, i.e., $\mathcal{C}_{\rm R-E}^{\rm
DPS'}(P)=\mathcal{C}_{\rm R-E}^{\rm OPS'}(P),P\geq 0$.
\end{proposition}
\begin{proof}
Please refer to Appendix \ref{appendix:proof prop2}.
\end{proof}

From Proposition \ref{proposition:2}, it suffices to consider
the OPS scheme for the optimal R-E tradeoff in the case of separated
receivers. Unlike the case of $\Pcs=0$,
where the boundary of $\mathcal{C}_{\rm R-E}^{\rm DPS}=\mathcal{C}_{\rm R-E}^{\rm SPS}$
is achieved as $\rho$ sweeps from 0 to 1,
the optimal power splitting pairs $(\alpha^{\ast},\rho^{\ast})$
that achieve the boundary of $\mathcal{C}_{\rm R-E}^{\rm DPS'}=\mathcal{C}_{\rm R-E}^{\rm OPS'}$
has to be determined. We thus consider the following optimization problem:
\begin{align}
\mathrm{(P0)}:~\mathop{\mathtt{max.}}_{\alpha,\rho} & ~~ R=(1-\alpha)\log_2 \left(1+\frac{(1-\rho)hP}{(1-\rho)\Va+\Vcov}\right) \nonumber \\
\mathtt{s.t.} & ~~ \alpha\zeta hP + (1-\alpha)\rho\zeta hP -(1-\alpha)\Pcs \geq Q,  \nonumber \\
              & ~~ 0\leq\alpha\leq1, ~ 0\leq\rho\leq1, \nonumber
\end{align}

Problem (P0) is feasible if and only if $Q\leq \zeta hP$. It is easy to verify that $(R,Q)=(0,\zeta hP)$
is achieved by $\alpha=1$. Next, we consider Problem (P0) for given $Q\in[0,\zeta hP)$ and $0\leq\alpha<1$.
The optimal solution of (P0) is obtained with the first constraint strictly equal,
otherwise we can always decrease $\alpha$ or $\rho$ to obtain a larger rate $R$.
Thus the boundary points $(R,Q)$ satisfy the following two equations,
\begin{align}
    Q&= \alpha\zeta hP + (1-\alpha)\rho\zeta hP -(1-\alpha)\Pcs,     \label{eq:Q-bound}\\
    R&= (1-\alpha)\log_2\left(1+\frac{(1-\rho)hP}{(1-\rho)\sigma^2_{\rm A}+\sigma^2_{\rm cov}} \right). \label{eq:R-bound}
\end{align}
From (\ref{eq:Q-bound}), we have
\begin{equation}\label{eq:rho-bound}
    \rho = \frac{Q-\alpha\zeta hP+(1-\alpha)\Pcs}{(1-\alpha)\zeta hP}.
\end{equation}
From (\ref{eq:rho-bound}), we have $\alpha\in\left[\max\{\frac{Q+\Pcs-\zeta hP}{\Pcs},0\},\frac{Q+\Pcs}{\zeta hP+\Pcs}\right]$ such that $0\leq\rho\leq1$.
Substituting (\ref{eq:rho-bound}) to (\ref{eq:R-bound}), we have
\begin{align}\label{eq:R-temp}
    R = (1-\alpha)\log_2\left(  1+
          \frac {\frac{\zeta hP-Q-(1-\alpha)\Pcs}{\zeta}} {\frac{\zeta hP-Q-(1-\alpha)\Pcs}{\zeta hP} \Va+\Vcov (1-\alpha)}
        \right).
\end{align}
From (\ref{eq:R-temp}), $R$ is a function of $\alpha$ with fixed $Q$. For convenience, we rewrite (\ref{eq:R-temp}) as follows:
\begin{equation}\label{eq:R(s)}
    R(s) = s\log_2\left( 1+\frac{cs+d}{as+b} \right)
\end{equation}
where $s=1-\alpha$, $a=\Vcov-\frac{\Va\Pcs}{\zeta hP}$, $b=\Va\left(1-\frac{Q}{\zeta hP}\right)>0$, $c=-\frac{\Pcs}{\zeta}<0$
and $d=hP\left(1-\frac{Q}{\zeta hP}\right)>0$.
It is worth noting that $s\in\left[\frac{\zeta hP-Q}{\zeta hP+\Pcs},\min\{\frac{\zeta hP-Q}{\Pcs},1\}\right]$, or equivalently, $s\in\left[\frac{d}{hP-c},\min\{-\frac{d}{c},1\}\right]$, since $\alpha\in\left[\max\{\frac{Q+\Pcs-\zeta hP}{\Pcs},0\},\frac{Q+\Pcs}{\zeta hP+\Pcs}\right]$.
The following lemma describes the behavior of $R(s)$ in terms of $s$, which is important for determining the boundary points $(R,Q)$.

\begin{lemma}\label{lemma:1}
With $Q\in [0,\zeta hP)$, $R(s)$ is concave in $s\in[\frac{d}{hP-c},\min{\{-\frac{d}{c},1\}}]$.
\end{lemma}
\begin{proof}
Please refer to Appendix \ref{appendix:proof lemma1}.
\end{proof}

By Lemma \ref{lemma:1}, the optimal $s^\ast\in[\frac{d}{hP-c},\min{\{-\frac{d}{c},1\}}]$ that maximizes $R(s)$ can be efficiently obtained by searching over $s\in[\frac{d}{hP-c},\min{\{-\frac{d}{c},1\}}]$ using the bisection method.
The optimal $\alpha^\ast$ is thus given by $\alpha^\ast=1-s^\ast$. The optimal $\rho^\ast$ is given by (\ref{eq:rho-bound}) with $\alpha=\alpha^{\ast}$. The corresponding $R$ is given by (\ref{eq:R-bound}) with $\alpha=\alpha^{\ast}$ and $\rho=\rho^{\ast}$. To summarize, each boundary point $(R,Q)$ of $\mathcal{C}_{\rm R-E}^{\rm OPS'}$ is achieved by a unique power splitting pair $(\alpha^{\ast},\rho^{\ast})$.

Fig. \ref{fig:RE-circuit-Sep} shows the achievable R-E regions (labeled as ``net energy'') for SepRx
with receiver circuit power consumption. The total harvested energy (labeled as ``total energy''),
including both the net energy
stored in the battery and the energy consumed by information decoding, is also shown in
Fig.~\ref{fig:RE-circuit-Sep} as a reference.
For SepRx with $\Pcs=25$, it is observed that
$\mathcal{C}_{\rm R-E}^{\rm TS'}\subseteq\mathcal{C}_{\rm R-E}^{\rm OPS'}$ and
$\mathcal{C}_{\rm R-E}^{\rm SPS'}\subseteq\mathcal{C}_{\rm R-E}^{\rm OPS'}$.
Moreover, SPS achieves the RE-region boundary only at low harvested energy region, where
$\mathcal{C}_{\rm R-E}^{\rm SPS'}$ and $\mathcal{C}_{\rm R-E}^{\rm OPS'}$
partially coincide. However, the performance of SPS becomes worse (even worse
than TS) when more harvested energy is desired, since it is unwise and energy-inefficient
to keep information receiver always on during the whole transmission time.

\begin{figure}
\centering
 \epsfxsize=0.7\linewidth
    \includegraphics[width=9cm]{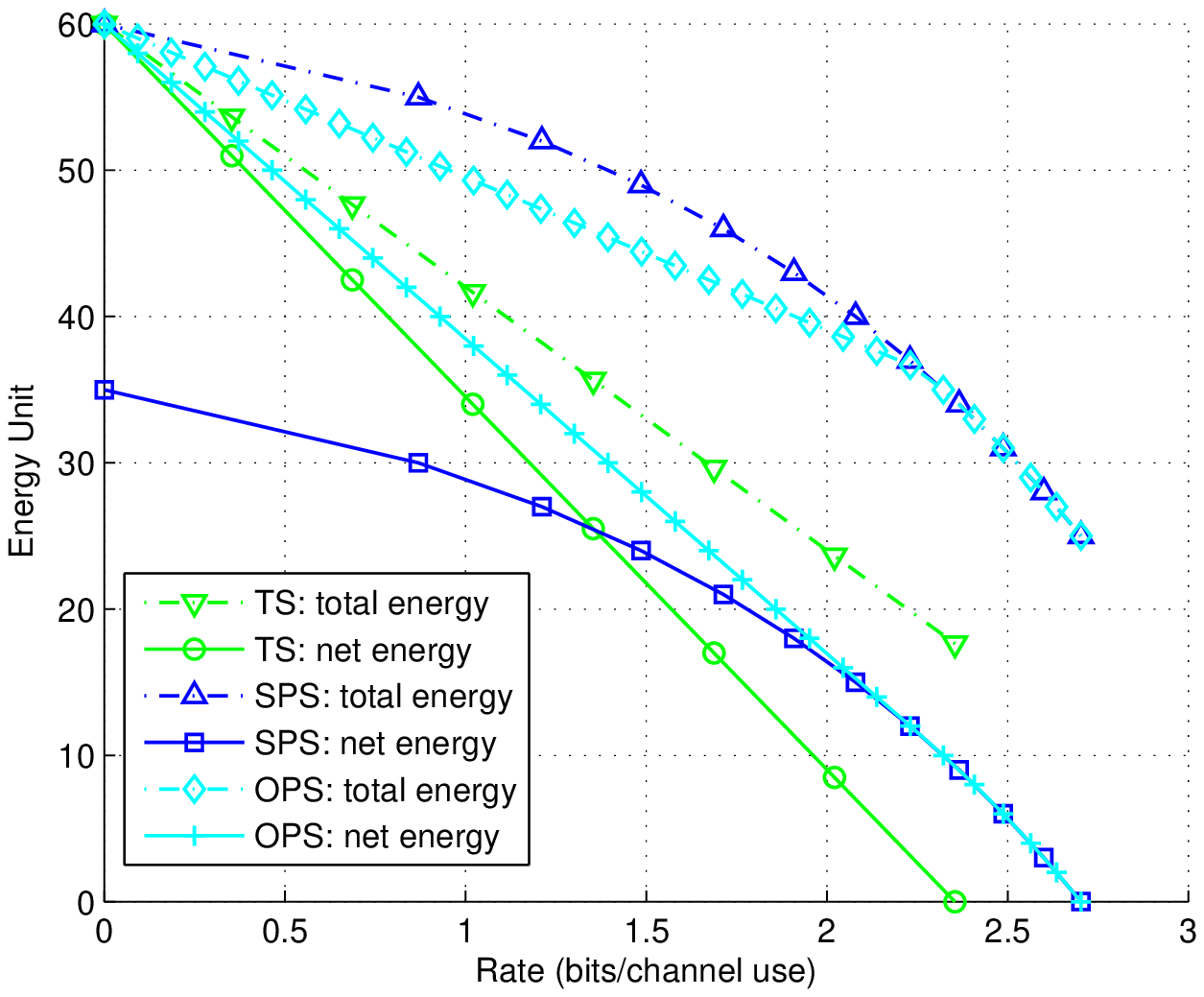}
\caption{Rate-energy tradeoff for the separated receiver with receiver circuit power consumption. It is assumed that $h=1, P=100, \zeta=0.6, \Va=1, \Vcov=10$ and $\Pcs=25$.}
\label{fig:RE-circuit-Sep}
\end{figure}

\subsection{Integrated Receiver with $\Pci>0$}
For the integrated receiver,
the achievable R-E region for the DPS scheme taking into account circuit power $\Pci$ is given by

\begin{footnotesize}\vspace{-0.1in}
\begin{align}\label{eq:RE Int-DPS-Pci}
    \mathcal{C}_{\rm R-E}^{\rm DPS'}
       (P)\triangleq \bigcup\limits_{\mv{\rho}}\bigg\{(R,Q):
&      0\leq Q\leq\frac{1}{N}\left(\sum\limits_{k=1}^{N}\rho_k \zeta hP -\sum\limits_{k=\alpha N+1}^{N}\Pci\right), \nonumber\\
&      R \left. \leq \frac{1}{N}\sum\limits_{k=\alpha N+1}^{N}C_{\rm NL} \right\}.
\end{align}
\end{footnotesize}

Since $R$ is independent of $\rho_k$, we should set $\rho_k\rightarrow 1$ for all $k=\alpha N+1,\ldots,N$.
Thus, the OPS scheme with $\rho\rightarrow 1$ is the optimal DPS scheme for the integrated receiver with $\Pci>0$.
Then (\ref{eq:RE Int-DPS-Pci}) can be simplified as
\begin{align}\label{eq:RE Int-OPS-Pci}
    \mathcal{C}_{\rm R-E}^{\rm OPS'}
       (P)\triangleq \bigcup\limits_{\alpha}\left\{(R,Q): \right.
&      0\leq Q\leq \zeta hP - (1-\alpha)\Pci, \nonumber\\
& \left.      R \leq (1-\alpha)C_{\rm NL} \right\}.
\end{align}
Note that when $\Pci<\zeta hP$, the boundary of $\mathcal{C}_{\rm R-E}^{\rm OPS'}(P)$ is determined by two lines
as $\alpha$ sweeps from 0 to 1, with one vertical line connecting
the two points $(C_{\rm NL},0)$ and $(C_{\rm NL},\zeta hP-\Pci)$, and another line
connecting the two points $(C_{\rm NL},\zeta hP-\Pci)$ and $(0,\zeta hP)$.
While $\Pci\geq\zeta hP$, the boundary of $\mathcal{C}_{\rm R-E}^{\rm OPS'}(P)$ is simply a straight
line connecting the two points $(\zeta hPC_{\rm NL}/\Pci,0)$ and $(0,\zeta hP)$ as $\alpha$ sweeps from
$1-\zeta hP/\Pci$ to $1$.

\begin{figure}
\centering
 \epsfxsize=0.7\linewidth
    \includegraphics[width=9cm]{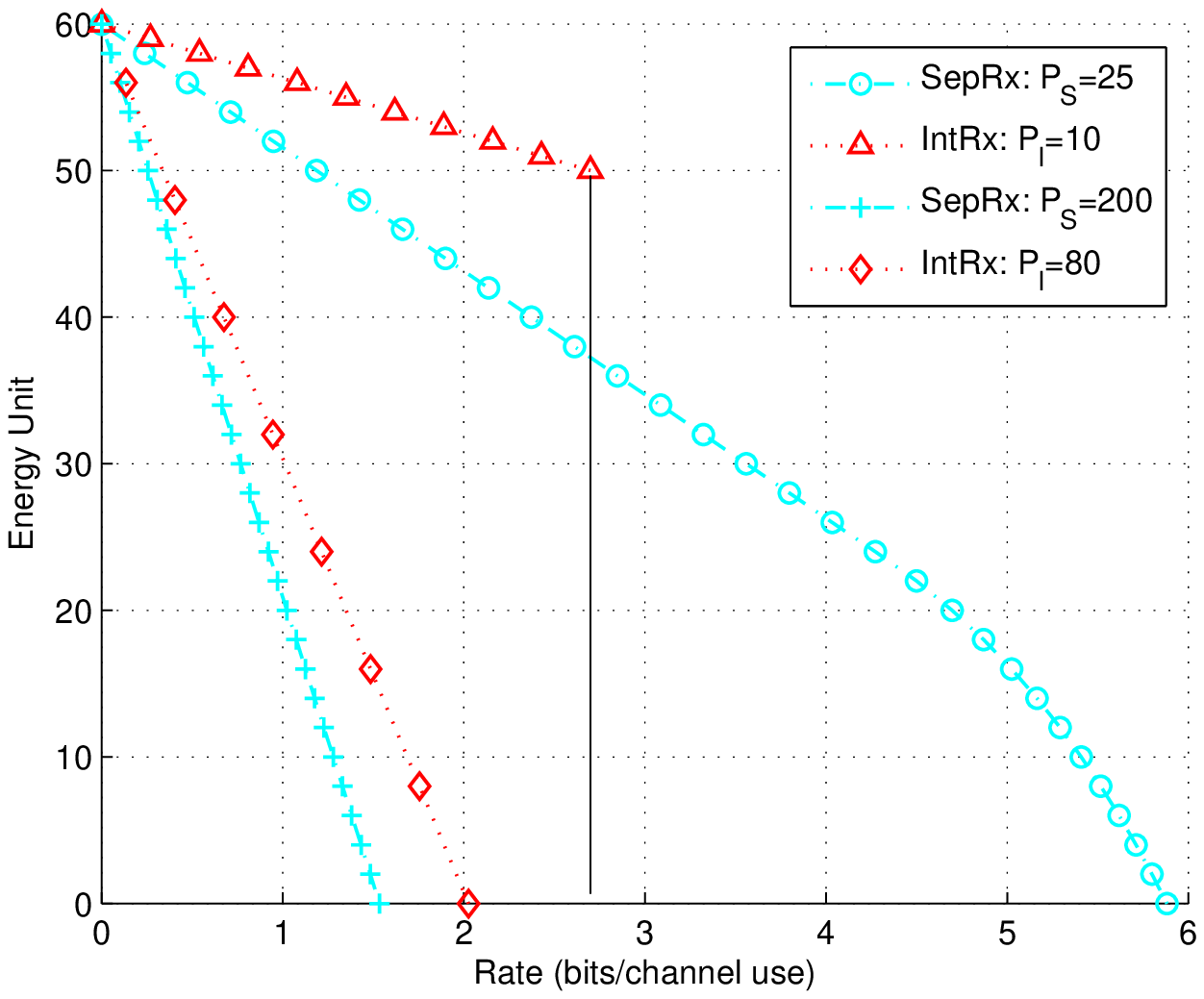}
\caption{Rate-energy tradeoff for separated vs. integrated receivers with receiver circuit power consumption. It is assumed that $h=1, P=100, \zeta=0.6, \Va=0.01, \Vcov=1$ and $\Srec=10$.}
\label{fig:RE-circuit-both}
\end{figure}

Fig. \ref{fig:RE-circuit-both} shows the achievable R-E regions for both cases of SepRx and IntRx
with receiver circuit power consumption. We consider two setups for the receiver circuit
power consumption, i.e., low circuit power with $\Pcs=25, \Pci=10$, and high circuit power
with $\Pcs=200, \Pci=80$.
For the low circuit power with $\Pcs=25, \Pci=10$, IntRx is superior over SepRx when more
harvested energy is desired, while SepRx is superior when less harvested energy
(no greater than 37 energy units) is required.
For the high circuit power with $\Pcs=200, \Pci=80$, IntRx is always superior over SepRx,
since for SepRx much more transmission time needs to be allocated for harvesting
energy to compensate the power consumed by information decoding.

\section{Practical Modulation}\label{sec:practical modulation}

In this section, we study the performances for the two types of receivers under
a realistic system setup that employs practical modulation.
Let the signal set (constellation) be denoted by $\mathcal{X}$.
The size of $\mathcal{X}$ is denoted by $M$ with $M=2^l$, and $l\geq 1$ being an integer.
It is assumed that the maximum rate that the practical modulation can support is $l\leq10$ bits/channel use.
The $i$-th constellation point in $\mathcal{X}$ is denoted by $x_i, i=1,\ldots,M$,
with equal probability $p_X(x_i)=1/M$ for simplicity.
For the separated receiver, we assume that coherent $M$-ary quadrature amplitude
modulation (QAM) is utilized for transmission. The symbol error rate (SER),
denoted by $P_{\rm s}^{\rm{QAM}}$, is approximated by \cite{Goldsmith}
\begin{equation}\label{eq:Ps-QAM}
    P_{\rm s}^{\rm{QAM}} \approx \frac{4(\sqrt{M}-1)}{\sqrt{M}} \mathcal{Q}\left( \sqrt{\frac{3\tau_s}{M-1}} \right)
\end{equation}
where $\tau_{\rm s}$ denotes the average SNR per symbol at the information receiver\footnote{Binary
phase shift keying (BPSK) is used when $l=1$.
For simplicity, we use (\ref{eq:Ps-QAM}) to approximate the SER of BPSK at high SNR.}.
The approximation is tight at high SNR, and is taken to be exact for simplicity in the sequel.
For the integrated receiver, as mentioned earlier in Section \ref{sec:integrated receivers},
information is encoded by the energy (power) of the transmitted
signal. Similar to the pulse amplitude modulation (PAM), we assume the {\it pulse energy modulation} (PEM), with equispaced {\it positive} constellation points given by
\begin{equation}\label{}
    x_i = \frac{2(i-1)}{M-1}, \ i=1,\ldots,M.
\end{equation}
A closed-form expression for the symbol error rate $P_{\rm s}^{\rm{PEM}}$ appears intractable, due
to the coupled antenna and rectifier noise for the channel (\ref{eq:new model}).
For most practical systems, the rectifier noise power will be much greater than the antenna noise power, while the antenna noise
is approximately at the thermal noise level. This justifies the assumption that $\Va\ll \Srec$ and we thus approximate the channel  (\ref{eq:new model})
with (\ref{eq:sim model-Na=0}). For simplicity, the decision boundary is chosen as the perpendicular bisector of each pair of adjacent two points, and the symbol error rate can be derived to be
\begin{equation}\label{}
    P_{\rm s}^{\rm{PEM}} = \frac{2(M-1)}{M} \mathcal{Q}\left( \frac{\tau_s'}{M-1} \right)
\end{equation}
where $\tau_s'=hP/\sigma_{\rm rec}$ is defined as the average SNR per symbol at the information receiver.

For both separated and integrated receivers, we assume the transmitter can adapt the transmission rate
such that the symbol error rate is less than a target value $P_{\rm s}^{\rm{tgt}}$,
i.e., $P_{\rm s}^{\rm{QAM}}\leq P_{\rm s}^{\rm{tgt}}$ and $P_{\rm s}^{\rm{PEM}}\leq P_{\rm s}^{\rm{tgt}}$ for the separated and integrated receivers, respectively.
Moreover, we assume that there is a minimum net harvested energy
requirement $Q_{\rm req}$ at the receiver side, i.e., $Q\geq Q_{\rm req}$, where $0\leq Q_{\rm req}\leq\zeta hP$.
With the SER constraint and minimum harvested energy constraint, our objective is to achieve the maximum rate.
For the separated receiver with OPS scheme, the maximum achievable rate can be obtained by

\begin{small}\vspace{-0.1in}
\begin{align}
\mathrm{(P1)}: \nonumber\\
~\mathop{\mathtt{max.}}_{\alpha,\rho,M} & ~~ R=(1-\alpha)\log_2 M \nonumber \\
\mathtt{s.t.} & ~~ \frac{4(\sqrt{M}-1)}{\sqrt{M}} \mathcal{Q}\left(\sqrt{\frac{3}{M-1}\cdot \frac{(1-\rho)hP}{(1-\rho)\Va+\Vcov}}\right)		
 \leq P_{\rm s}^{\rm tgt}, \label{eq:SER-P1} \\
              & ~~ \alpha\zeta hP + (1-\alpha)\rho\zeta hP -(1-\alpha)\Pcs \geq Q_{\rm req},  \label{eq:energy-P1} \\
              & ~~ 0\leq\alpha\leq1, ~ 0\leq\rho\leq1, \nonumber \\
              & ~~ M=2^{l}, ~ l\in \{1,2,\ldots,10\} \nonumber
\end{align}
\end{small}Here, the optimization variables are the power splitting pair $(\alpha,\rho)$ and the modulation size $M$.

For the integrated receiver with OPS scheme, the maximum achievable rate can be obtained by
\begin{align}
\mathrm{(P2)}:~\mathop{\mathtt{max.}}_{\alpha,M} & ~~ R=(1-\alpha)\log_2 M \nonumber \\
\mathtt{s.t.} & ~~ \frac{2(M-1)}{M} \mathcal{Q}\left(\frac{1}{M-1}\cdot\frac{hP}{\sigma_{\rm rec}}\right) \leq P_{\rm s}^{\rm tgt}, \label{eq:SER-P2}\\
              & ~~ \zeta hP -(1-\alpha)\Pci \geq Q_{\rm req},  \nonumber \\
              & ~~ 0\leq\alpha\leq1, \nonumber \\
              & ~~ M=2^{l}, ~ l\in \{1,2,\ldots,10\} \nonumber
\end{align}
Note that here the optimization variables only include $\alpha$ and $M$, since the OPS scheme with
$\rho\rightarrow 1$ is optimal for the integrated receiver (c.f. Section \ref{sec:circuit power}).

We denote the maximum rate for (P1) and (P2) as $\Rs$ and $\Ri$, respectively. Similarly, the optimal
variables for (P1) and (P2) are denoted with corresponding superscripts and subscripts,
e.g., $\alpha_1^\ast, \rho_1^\ast$, etc.
With $0\leq Q_{\rm req}\leq\zeta hP$ and reasonable SNR (such that the SER constraints can be satisfied by some $M$),
the optimal solution for (P1) is obtained by an exhaustive search for $\rho_1^\ast$: for each fixed $\rho_1\in[0,1)$, we have $R_1^\ast=(1-\alpha_1^\ast)\log_2 M_1^\ast$, where $\alpha_1^\ast=\left[\frac{Q_{\rm req}-\rho_1\zeta hP+\Pcs}{(1-\rho_1)\zeta hP+\Pcs}\right]^+$ and $M_1^\ast$ attains the maximum value under the SER constraint (\ref{eq:SER-P1}); the optimal $\rho_1^\ast$ is then obtained to maximize $R_1^\ast$.
The optimal solution for (P2) is given by $\Ri=(1-\alpha_2^\ast)\log_2 M_2^\ast$, where
$\alpha_2^\ast=\left[\frac{Q_{\rm req}-\zeta hP+\Pci}{\Pci}\right]^+$ and $M_2^\ast$ is maximized under
the SER constraint (\ref{eq:SER-P2}).
For both (P1) and (P2), the achievable rate $R$ is determined
by both the modulation size $M$ and the time percentage $\alpha$ that the information decoder
operates in the off mode. Moreover, as the received signal power $hP$ decreases, $M$ decreases
to satisfy the modulation constraint and $\alpha$ increases to satisfy
the harvested energy constraint, both of which result in a decrease of the achievable rate.

Typically for practical systems we have $\Pcs>\Pci>0$, since the RF band mixer in the separated receiver
will consume additional circuit power. Henceforth, we assume $\Pcs>\Pci>0$.
\begin{proposition}\label{proposition:3}
For separated and integrated receivers with $0\leq Q_{\rm req}\leq \zeta hP$ and $\Pcs\geq\Pci>0$, we have $\alpha_1^\ast\geq\alpha_2^\ast$. Moreover,
if $M^\ast_1\leq M_2^\ast$, then the maximum achievable rate by the separated receiver will be
no greater than that by the integrated receiver, i.e., $\Rs\leq \Ri$.
\end{proposition}
\begin{proof}
Please refer to Appendix \ref{appendix:proof prop3}.
\end{proof}

Most practical systems of interest typically operate at the
high SNR regime for the information receiver, due to the high-power operating requirement for the energy
receiver. Thus, for sufficiently small transmission distance, it is expected that
both receivers can support the maximum modulation size under the SER constraint, i.e., $M^\ast_1= M_2^\ast=2^{10}$.
Thus, by Proposition \ref{proposition:3}, the integrated receiver
outperforms the separated receiver for sufficiently small transmission distance.

\begin{figure}
\centering
 \epsfxsize=0.7\linewidth
    \includegraphics[width=9cm]{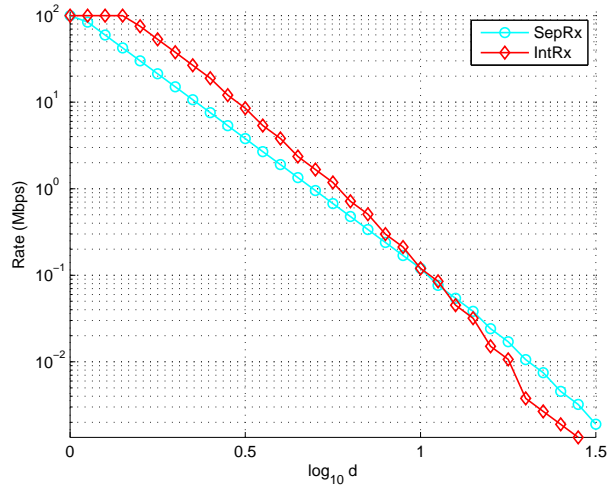}
\caption{Maximum achievable rate for separated and integrated receivers over different transmission distance.}
\label{fig:practical compare rate}
\end{figure}

Fig. \ref{fig:practical compare rate} shows an example of the maximum achievable rate
for a practical point-to-point wireless system with separated or integrated receiver.
The corresponding modulation size $M$ and time percentage $\alpha$ are shown in Fig. \ref{fig:L and alpha}.
The transmitter power is assumed to be $P=1$ watt(W) or $30\rm{dBm}$.
The distance from the transmitter to the receiver
is assumed to be $d$ meters with $d\geq1$, which results in approximately $(-30-30\log_{10}d)\rm{dB}$ of signal
power attenuation at a carrier frequency assumed as $f_{\rm c}=900\rm{MHz}$.
The bandwidth of the transmitted signal is assumed to be $10\rm{MHz}$.
For information receiver, the antenna noise temperature is assumed to be $290\rm{K}$, which
corresponds to $\Va=-104\rm{dBm}$ over the bandwidth of $10\rm{MHz}$.
As in most practical wireless communication systems, it is assumed that the processing noise
power is much greater than the antenna noise power, in which case the antenna noise can be omitted.
In particular, it is assumed that $\Vcov=-70\rm{dBm}$ for the separated receiver \cite{TI-noise} and $\Srec=-50\rm{dBm}$
for the integrated receiver.
The circuit power consumed by information decoding is assumed to be $\Pcs=0.5\rm{mW}$ for the separated
receiver, and $\Pci=0.2\rm{mW}$ for the integrated receiver.
For energy receiver, the energy conversion efficiency is assumed to be $\zeta=60\%$.
The minimum harvested energy requirement $Q_{\rm req}$ is set to be zero, which is the minimum requirement
for a zero-net-energy system that does not need external power source, i.e., the receiver is ``self-sustainable''.
The symbol error rate target is assumed to be $P_{\rm s}^{\rm{tgt}}=10^{-5}$.

\begin{figure}
\centering
 \epsfxsize=0.7\linewidth
    \includegraphics[width=9cm]{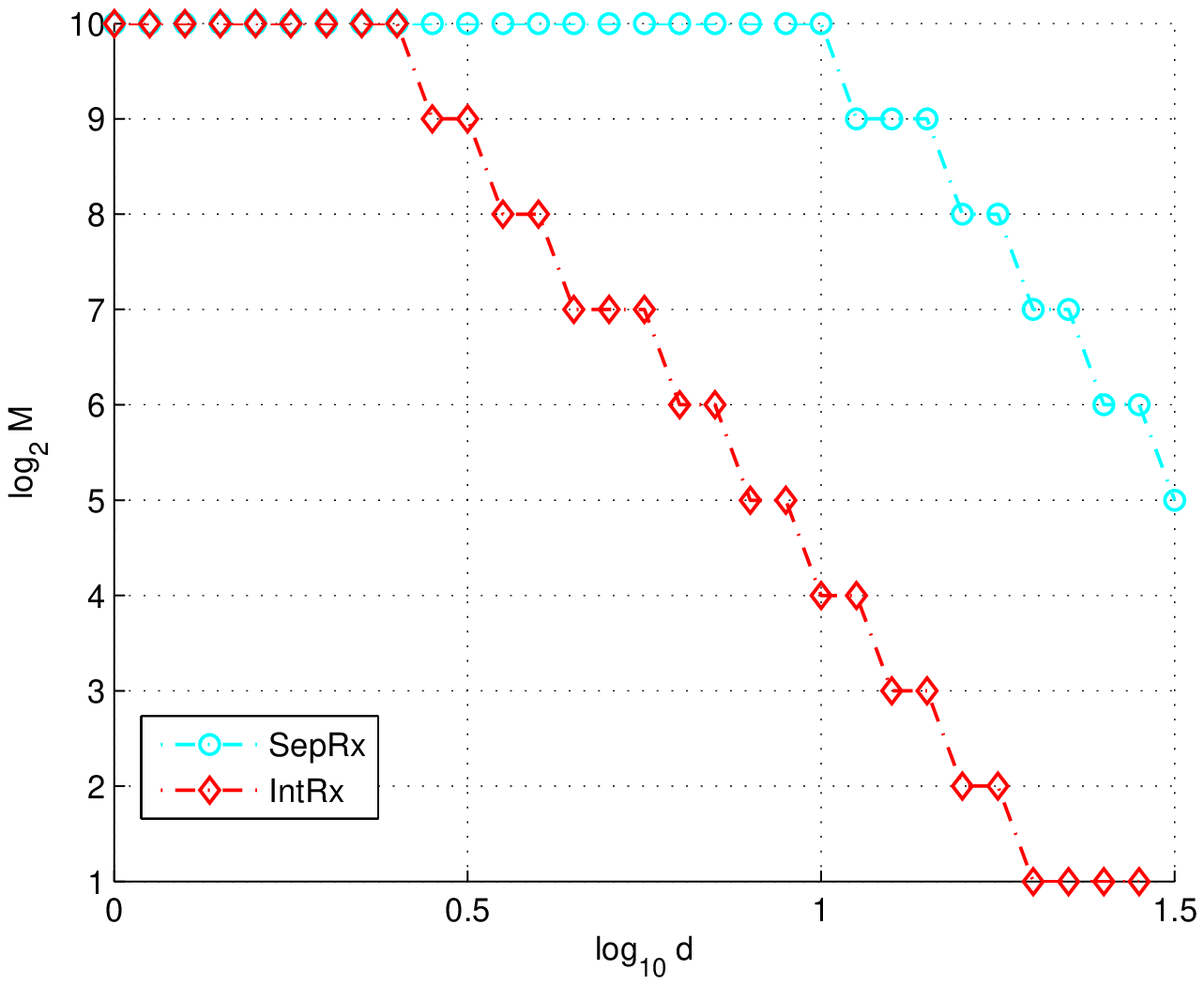}
    \includegraphics[width=9cm]{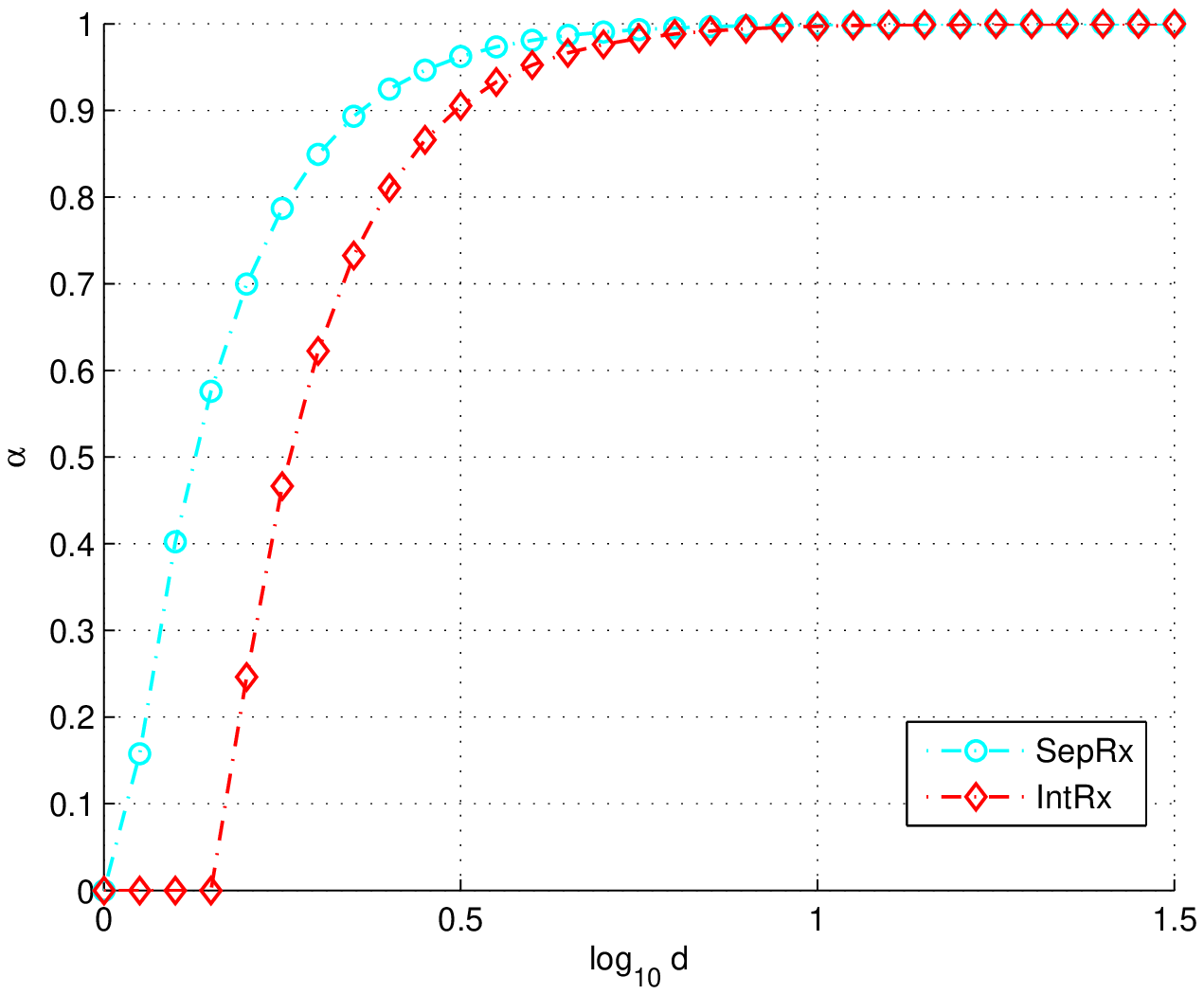}
\caption{Optimal modulation size ($M$) and information receiver off-time percentage ($\alpha$) for separated and integrated receivers.}
\label{fig:L and alpha}
\end{figure}

In Fig. \ref{fig:practical compare rate}, it is observed that
when $0\leq \log_{10} d \leq 1$, IntRx achieves more rate than SepRx.
By Proposition \ref{proposition:3}, IntRx outperforms SepRx over the range $0\leq\log_{10} d\leq0.4$ with $M_1^\ast=M_2^\ast=2^{10}$;
however, Proposition \ref{proposition:3} provides only a sufficient condition, numerical results
show that IntRx outperforms SepRx over longer distances up to $\log_{10} d\leq1$. This is due to the fact that
although SepRx supports higher-order constellations (larger $M$) than IntRx
when $0.4<\log_{10} d \leq 1$, the information receiver of SepRx needs to operate in the off mode
for more time (larger $\alpha$) to compensate the power consumed by information decoding
(c.f. Fig. \ref{fig:L and alpha}). It turns out that over this range, the average rate over the whole transmission time
of SepRx is less than that achieved by IntRx.
As $\log_{10}d$ increases, the rate gap between SepRx and IntRx shrinks and converges when $\log_{10}d$ is around 1.
When $1.1\leq \log_{10} d \leq 1.5$, SepRx achieves more rate than IntRx, since $\alpha$ for both receivers
approaches to 1 (c.f. Fig. \ref{fig:L and alpha}), while the achievable rates are dominated by the modulation size
($M$, c.f. Fig. \ref{fig:L and alpha}).
Note that when $\log_{10}d=1.5$, no modulation can support IntRx due to the extremely low received SNR; however, SepRx
can still achieve some positive rate.
In addition, Fig. \ref{fig:L and alpha} shows that in general IntRx exploits lower complexity (smaller $M$) in generating signal constellation.

\section{Conclusion}\label{sec:conclusion}
This paper investigates practical receiver designs for simultaneous wireless
information and power transfer. Based on {\it dynamic power splitting} (DPS),
we propose two practical receiver architectures, namely, {\it separated} and
{\it integrated} information and energy receivers. For the separated receiver,
the received signal by the antenna is split into two signal streams in the RF band,
which are then separately fed to the conventional energy receiver and information
receiver for harvesting energy and decoding information, respectively.
For the integrated receiver, part of the information decoding implementation,
i.e., the RF to baseband conversion, is integrated to the energy receiver via the
rectifier. For both receivers, we characterize the rate-energy performance taking
circuit power consumption into account. Numerical results show that
when the circuit power consumptions are small (compared with the received signal power),
the separated receiver is superior at low harvested energy region;
whereas the integrated receiver performs better at high harvested energy region.
When the circuit power consumptions are large, the integrated receiver is superior.
Moreover, the performance for the two types of receivers is studied under a realistic system
setup that employs practical modulation. With symbol error rate constraint and minimum
harvested energy constraint, the maximum achievable rates by the two types
of receivers are compared. It is shown that for a system with zero-net-energy consumption,
the integrated receiver achieves more rate than separated receiver at small transmission distances.


\appendices

\section{Proof of Proposition \ref{proposition:1}}\label{appendix:proof prop1}
To show $\mathcal{C}_{\rm R-E}^{\rm DPS}(P)=\mathcal{C}_{\rm R-E}^{\rm SPS}(P),P\geq 0$,
it suffices for us to show that $\mathcal{C}_{\rm R-E}^{\rm SPS}(P)\subseteq\mathcal{C}_{\rm R-E}^{\rm DPS}(P),P\geq 0$
and $\mathcal{C}_{\rm R-E}^{\rm DPS}(P)\subseteq\mathcal{C}_{\rm R-E}^{\rm SPS}(P),P\geq 0$.
The first part of proof is trivial, since SPS is just a special case of DPS by
letting $\rho_k=\rho, \forall k$ (c.f. (\ref{eq:rho-SPS})).
Next, we prove the second part.
Assuming that $f(\rho)=\log_2 \left(1+\frac{(1-\rho)hP}{(1-\rho)\sigma^2_{\rm A}+\sigma^2_{\rm cov}}\right)$,
it is easy to verify that $f(\rho)$ is concave in $\rho\in[0,1]$. By Jensen's inequality, we have
$\frac{1}{N}\sum\limits_{k=1}^{N} f(\rho_k) \leq f\left(\frac{1}{N}\sum\limits_{k=1}^{N}\rho_k\right)$.
Thus, for $\forall {\mv{\rho}}=[\rho_1,\ldots,\rho_N]^T$,
$\exists \rho=\frac{1}{N}\sum\limits_{k=1}^{N}\rho_k$, so that $\frac{1}{N}\sum\limits_{k=1}^{N}\rho_k \zeta hP
= \rho \zeta hP$ and $\frac{1}{N}\sum\limits_{k=1}^{N} f(\rho_k) \leq f(\rho)$.
Since R-E region is defined as the union of rate-energy pairs $(R,Q)$ under all possible $\mv{\rho}$, it follows
immediately that $\mathcal{C}_{\rm R-E}^{\rm DPS}(P)\subseteq\mathcal{C}_{\rm R-E}^{\rm SPS}(P),P\geq 0$,
which completes the proof of Proposition \ref{proposition:1}.

\section{Proof of Proposition \ref{proposition:2}}\label{appendix:proof prop2}
To show $\mathcal{C}_{\rm R-E}^{\rm DPS'}(P)=\mathcal{C}_{\rm R-E}^{\rm OPS'}(P),P\geq 0$,
it suffices for us to show that $\mathcal{C}_{\rm R-E}^{\rm OPS'}(P)\subseteq\mathcal{C}_{\rm R-E}^{\rm DPS'}(P),P\geq 0$
and $\mathcal{C}_{\rm R-E}^{\rm DPS'}(P)\subseteq\mathcal{C}_{\rm R-E}^{\rm OPS'}(P),P\geq 0$.
The first part of proof is trivial, since OPS is just a special case of DPS by
letting $\rho_k=\rho, k=\alpha N,\ldots,N$ (c.f. (\ref{eq:rho-OPS})).
Next, we prove the second part.
By (\ref{eq:RE DPS-Pmix}), $\rho_k$'s are optimized at $\rho_k=1$ for $k=1,\ldots,\alpha N$; thus, we have
$Q\leq\alpha\zeta hP+\frac{1}{N}\sum\limits_{k=\alpha N+1}^{N}\rho_k \zeta hP -(1-\alpha)\Pcs$ for DPS.
For any given $\alpha$, by Jensen's inequality we have
$\frac{1}{(1-\alpha)N}\sum\limits_{k=\alpha N+1}^{N} f(\rho_k) \leq f\left(\frac{1}{(1-\alpha)N}\sum\limits_{k=\alpha N+1}^{N}\rho_k\right)$. Thus, for $\forall \alpha$ and $\forall {\mv{\rho}}=[\rho_1,\ldots,\rho_N]^T$,
$\exists \rho=\frac{1}{(1-\alpha)N}\sum\limits_{k=\alpha N+1}^{N}\rho_k$, so that
$\frac{1}{N}\sum\limits_{k=\alpha N+1}^{N}\rho_k \zeta hP = (1-\alpha)\rho\zeta hP$ and
$\frac{1}{N}\sum\limits_{k=\alpha N+1}^{N} f(\rho_k) \leq (1-\alpha)f(\rho)$.
Since R-E region is defined as the union of rate-energy pairs $(R,Q)$ under all possible $\mv{\rho}$, it follows
immediately that $\mathcal{C}_{\rm R-E}^{\rm DPS'}(P)\subseteq\mathcal{C}_{\rm R-E}^{\rm OPS'}(P),P\geq 0$,
which completes the proof of Proposition \ref{proposition:2}.

\section{Proof of Lemma \ref{lemma:1}}\label{appendix:proof lemma1}
From (\ref{eq:R(s)}), the first and second derivatives of $R(s)$ with respect of $s$ are given by
\begin{equation}\label{eq:R'}
    \frac{d R}{d s}=\log_2\left(1+\frac{cs+d}{as+b}\right) + \frac{s(bc-ad)}{\left((a+c)s+b+d\right)(as+b)\ln 2},
\end{equation}
\begin{equation}\label{eq:R''}
    \frac{d^2 R}{d s^2}=\frac{(bc-ad)\left(\left(b(a+c)+a(d+d)\right)s+2b(d+d)\right)}{\left((a+c)s+b+d\right)^2(as+b)^2\ln 2}.
\end{equation}
From (\ref{eq:R''}), the sign of $\frac{d^2 R}{d s^2}$ is identical with the line
$f_2(s)=(bc-ad)\left(\left(b(a+c)+a(d+d)\right)s+2b(d+d)\right)$. Note that $bc-ad=-\Vcov(hP-Q/\zeta)< 0$,
$f_2(0)=2b(b+d)(bc-ad)< 0$, and $f_2(-\frac{d}{c})=\frac{(2b+d)(bc-ad)^2}{c}< 0$; thus, we have
$\frac{d^2 R}{d s^2}< 0$ for $s\in [0,-\frac{d}{c}]$. Since the set $[\frac{d}{hP-c},\min{\{-\frac{d}{c},1\}}]$
is a subset of the set $[0,-\frac{d}{c}]$, we have $\frac{d^2 R}{d s^2}< 0$ for $s\in[\frac{d}{hP-c},\min{\{-\frac{d}{c},1\}}]$. Thus,
$R(s)$ is concave in $s\in [\frac{d}{hP-c},\min{\{-\frac{d}{c},1\}}]$, which completes the proof of Lemma \ref{lemma:1}.

\section{Proof of Proposition \ref{proposition:3}}\label{appendix:proof prop3}
We first consider (P1) with $0\leq Q_{\rm req}\leq \zeta hP$ for the separated receiver.
The optimal $\alpha$ for (P1) is given by $\alpha_1^\ast=\left[\frac{Q_{\rm req}-\rho_1^\ast\zeta hP+\Pcs}{(1-\rho_1^\ast)\zeta hP+\Pcs}\right]^+$.
Since $Q_{\rm req}\leq \zeta hP$, $\alpha^\ast_1$ decreases as $\rho_1^\ast$ increases.
Thus, we have
\begin{equation}\label{eq:temp-alpha-Sep-2}
    \alpha_1^\ast\geq\frac{Q_{\rm req}-\rho_1^\ast\zeta hP+\Pcs}{(1-\rho_1^\ast)\zeta hP+\Pcs}\Big|_{\rho_1^\ast=1}=\frac{Q_{\rm req}-\zeta hP+\Pcs}{\Pcs}.
\end{equation}

Next, for the integrated receiver with $0\leq Q_{\rm req}\leq \zeta hP$, the optimal $\alpha$ for (P2) is given by
\begin{equation}\label{eq:temp-alpha-Int}
    \alpha_2^\ast=\left[\frac{Q_{\rm req}-\zeta hP+\Pci}{\Pci}\right]^+
\end{equation}
From (\ref{eq:temp-alpha-Sep-2}) and (\ref{eq:temp-alpha-Int}), we have $\alpha_1^\ast\geq\alpha_2^\ast$, given that $\Pcs>\Pci$. Since $R=(1-\alpha)\log_2 M$, we have $R_1^\ast\leq R_2^\ast$, given that $\alpha_1^\ast\geq\alpha_2^\ast$ and $M_1^\ast\leq M_2^\ast$. The proof of Proposition \ref{proposition:3} thus follows.


\begin{thebibliography}{1}
\bibliographystyle{IEEEbib}

\bibitem{Varshney} L. R. Varshney, ``Transporting information and energy simultaneously,'' in {\it Proc. IEEE Int. Symp. Inf. Theory (ISIT)}, pp. 1612-1616, July 2008.

\bibitem{Sahai} P. Grover and A. Sahai, ``Shannon meets Tesla: wireless information and power transfer,'' in {\it Proc. IEEE Int. Symp. Inf. Theory (ISIT)}, pp. 2363-2367, June 2010.

\bibitem{Liu} L. Liu, R. Zhang, and K. C. Chua, ``Wireless information transfer with opportunistic energy harvesting,'' {\it IEEE Trans. Wireless Commun.}, vol. 12, no. 1, pp. 288-300, Jan. 2013.

\bibitem{Zhang} R. Zhang and C. K. Ho, ``MIMO broadcasting for simultaneous wireless information and power transfer,'' {\it IEEE Trans. Wireless Commun.}, vol. 12, no. 5, pp. 1989-2001, May 2013.

\bibitem{Xiang} Z. Xiang and M. Tao, ``Robust beamforming for wireless information and power transmission,'' {\it IEEE Wireless Commun. Letters}, vol. 1, no. 4, pp. 372-375, 2012.

\bibitem{Chalise} B. K. Chalise, Y. D. Zhang, and M. G. Amin, ``Energy harvesting in an OSTBC based amplify-and-forward MIMO relay system,'' in {\it Proc. IEEE ICASSP}, pp. 3201-3204, Mar. 2012.

\bibitem{Fouladgar} A. M. Fouladgar and O. Simeone, ``On the transfer of information and energy in multi-user systems,'' {\it IEEE Commun. Letters}, vol. 16, no. 11, pp. 1733-1736, Nov. 2012.

\bibitem{Huang} K. Huang and V. K. N. Lau, ``Enabling wireless power transfer in cellular networks: architecture, modeling and deployment,'' available on-line at arXiv:1207.5640.

\bibitem{Lee}  S. Lee, R. Zhang, and K. Huang, ``Opportunistic wireless energy harvesting in cognitive radio networks,'' to appear in {\it IEEE Trans. Wireless Commun.}, available on-line at arXiv:1302.4793.

\bibitem{Powerdivider} Y. Wu, Y. Liu, Q. Xue, S. Li, and C. Yu,``Analytical design method of multiway dual-band planar power dividers with arbitrary power division,'' {\it IEEE Trans. Microwave Theory and Techniques}, vol. 58, no. 12, pp. 3832-3841, Dec 2010.

\bibitem{Agilent} Product Datasheet, 11667A Power Splitter, Agilent Technologies.


\bibitem{Cover} T. Cover and J. Thomas, {\it Elements of information theory}, New York: Wiley, 1991.

\bibitem{Paing} T. Paing, J. Shin, R. Zane, and Z. Popovic, ``Resistor emulation approach to low-power RF energy harvesting,'' {\it IEEE Trans. Power Electronics}, vol. 23, no. 3, pp. 1494-1501, May 2008.

\bibitem{Akkermans} J. A. G. Akkermans, M. C. van Beurden, G. J. N. Doodeman, and H. J. Visser, ``Analytical models for low-power rectenna design,'' {\it IEEE Antennas Wireless Propag. Letters}, vol. 4, pp. 187-190, 2005.

\bibitem{Powercast} Product Datasheet, P2110-915MHz RF Powerharvester Receiver, Powercast Corporation.


\bibitem{Urkowitz} H. Urkowitz, ``Energy detection of unknown deterministic signals,'' {\it Proc. IEEE}, vol. 55, pp. 523-231, Apr. 1967.

\bibitem{Faycal} I. Abou-Faycal and J. Fahs, ``On the capacity of some deterministic non-linear channels subject to additive white Gaussian noise,'' in {\it Proc. IEEE Int. Conf. on Telecommunications (ICT)}, pp. 63-70, Apr. 2010.

\bibitem{Lapidoth-Moser} A. Lapidoth, S. M. Moser, and M. A. Wigger, ``On the capacity of free-space optical intensity channels,'' {\it IEEE Trans. Inf. Theory,} vol. 55, no. 10, pp. 4449-4461, Oct. 2009.

\bibitem{Shamai} M. Katz and S. Shamai (Shitz), ``On the capacity-achieving distribution of the discrete-time noncoherent and partially coherent AWGN channels,'' {\it IEEE Trans. Inf. Theory}, vol. 50, no. 10, pp. 2257-2270, Oct. 2004.

\bibitem{Lapidoth} A. Lapidoth, ``On phase noise channels at high SNR,'' in {\it Proc. IEEE Inf. Theory Workshop (ITW)}, Oct. 2002.

\bibitem{Carbone} P. Carbone and D. Petri,``Noise sensitivity of the ADC histogram test,'' {\it IEEE Trans. Instrum. Meas.}, vol. 47, no. 4, pp. 1001-1004, Aug. 1998.

\bibitem{Ruscak} S. Ruscak and L. Singer, ``Using histogram techniques to measure A/D converter noise,'' {\it Analog Dialogue}, vol. 29, no. 2, 1995.

\bibitem{Goldsmith} A. Goldsmith, {\it Wireless communications}, Cambridge University Press, 2005.

\bibitem{TI-noise} M. Loy, ``Understanding and enhancing sensitivity in receivers for wireless applications,'' Technical brief SWRA030, Texas Instruments, available online at http://www.ti.com/lit/an/swra030/swra030.pdf

\end{thebibliography}
\end{document}